\documentclass[10pt,conference]{IEEEtran}
\pdfoutput=1
\usepackage{diagbox}
\usepackage{graphicx}
\usepackage{amsmath}
\usepackage{amsfonts}
\usepackage{algpseudocode}
\usepackage{algorithm}
\usepackage{subfigure}
\usepackage{tikz}
\usepackage{epsfig}

\usetikzlibrary{shapes}
\usetikzlibrary{arrows,automata}
\usetikzlibrary{positioning}
\usetikzlibrary{patterns}
\usetikzlibrary{backgrounds}
\newtheorem{theorem}{Theorem}

\newtheorem{definition}{Definition}
\newtheorem{lemma}{Lemma}

\renewcommand{\ae}[1]{{\color{black}{#1}}}

\begin{document}

\title{Resource Efficient Isolation Mechanisms in Mixed-Criticality Scheduling}
\author{
\IEEEauthorblockN{Xiaozhe Gu, Arvind Easwaran}
\IEEEauthorblockA{Nanyang Technological University, Singapore \\
Email: guxi0002@e.ntu.edu.sg, arvinde@ntu.edu.sg}
\and
\IEEEauthorblockN{Kieu-My Phan, Insik Shin}
\IEEEauthorblockA{KAIST, Korea \\
Email: phankieumy@kaist.ac.kr, insik.shin@cs.kaist.ac.kr}
}

\maketitle

\begin{abstract}
Mixed-criticality real-time scheduling has been developed to improve resource utilization while guaranteeing safe execution of critical applications. These studies use optimistic resource reservation for all the applications to improve utilization, but prioritize critical applications when the reservations become insufficient at runtime. Many of them however share an impractical assumption that all the critical applications will simultaneously demand additional resources. As a consequence, they under-utilize resources by penalizing all the low-criticality applications. In this paper we overcome this shortcoming using a novel mechanism that comprises a parameter to model the expected number of critical applications simultaneously demanding more resources, and an execution strategy based on the parameter to improve resource utilization. Since most mixed-criticality systems in practice are component-based, we design our mechanism such that the component boundaries provide the isolation necessary to support the execution of low-criticality applications, and at the same time protect the critical ones. We also develop schedulability tests for the proposed mechanism under both a flat as well as a hierarchical scheduling framework. Finally, through simulations, we compare the performance of the proposed approach with existing studies in terms of schedulability and the capability to support low-criticality applications.

\end{abstract}

\section{Introduction}
\label{sec:introduction}

An increasing trend in embedded systems is towards open computing environments, where multiple functionalities are developed independently and integrated together on a single computing platform~\cite{prisaznuk1992integrated}. 
An important notion behind this trend is the safe isolation of separate functionalities, primarily to achieve fault containment. 
This raises the challenge of how to balance the conflicting requirements of isolation for safety assurance and efficient resource sharing for economical benefits. The concept of \textit{mixed-criticality} appears to be important in meeting those two goals.

In many safety-critical systems, the correct behavior of some functionality (e.g., flight control) is more important (``critical'') to the overall safety of the system than that of another (e.g., in-flight cooling). In order to certify such systems as being correct, they are conventionally assessed under certain assumptions on the worst-case run-time behavior. For example, the estimation of Worst-Case Execution Times (WCETs) of code for highly critical functionalities involves very conservative assumptions that are unlikely to occur in practice. Such assumptions make sure that the resources reserved for critical functionalities are always sufficient. Thus, the system can be designed to be fully safe from a certification perspective, but the resources are in fact severely under-utilized in practice.

In order to close such a gap in resource utilization, Vestal~\cite{Ves07} proposed the mixed-criticality task model that comprises of different WCET values. These different values are determined at different levels of confidence (``criticality'') based on the following principle. A reasonable low-confidence WCET estimate, even if it is based on measurements, may be sufficient for almost all possible execution scenarios in practice. In the highly unlikely event that this estimate is violated, as long as the scheduling mechanism can ensure deadline satisfaction for highly critical applications, the resulting system design may still be considered as safe.

To ensure deadline satisfaction of critical applications, mixed-criticality studies make pessimistic assumptions when a single high-criticality task executes beyond its expected (low-confidence) WCET. They assume that the system will either immediately ignore all the low-criticality tasks~\cite{BBD11,BaFo11,GES11,BBA12,EkYi12,Eas13} or degrade the service offered to them~\cite{BuBa13,JZP13,Su13,PCH14}. They further assume that all the high-criticality tasks in the system can thereafter request for additional resources, up to their pessimistic (high-confidence) WCET estimates. Although these strategies ensure safe execution of critical applications, they have a serious drawback as pointed out in a recent article~\cite{BuBa13}. When a high-criticality task exceeds its expected WCET, the likelihood that all the other high-criticality tasks in the system will also require more resources is very low in practice. For instance, it is unlikely that adaptive cruise control and anti-lock braking, both of which are critical, would simultaneously require additional resources because their execution time depends on different inputs. Cruise control would most likely require additional resources when the cameras and lidars provide dense data, whereas the execution of anti-lock braking mainly depends on speed of the vehicle and friction on the tyres. Therefore, to penalize all the low-criticality tasks in the event that some high-criticality tasks require additional resources seems unreasonable.

In practice, most mixed-criticality systems are component-based wherein different vendors independently design and develop the various applications. For wide applicability, it is then natural that mixed-criticality scheduling strategies must consider the impact of WCET violations across component boundaries. To the extent possible, these strategies must limit this impact to within components, so that other components in the system can continue their execution uninterrupted. One extreme manifestation of this view is the reservation-based approach that completely isolates components but severly under-utilizes the resources. On the other hand, most of the recent mixed-criticality studies such as those mentioned above, completely ignore these component boundaries but still under-utilize resources due to unrealistic assumptions.

\textbf{Contributions.} Addressing the two central issues described above, in this paper we propose a resource efficient mechanism to support low-criticality tasks while still ensuring isolation of high-criticality tasks. This mechanism comprises the following.
\begin{enumerate}
\item A new parameter to model the expected number of simultaneous violation of low-confidence WCET by high-criticality tasks.
\item A corresponding execution strategy that maximizes low-criticality task executions as long as this number is not exceeded.
\item To efficiently support component-based mixed-criticality systems, we employ our mechanism at the component level. We ensure that as long as the number of low-confidence WCET violations within a component does not exceed the component's expected limit, task executions in other components, including low-criticality ones, remain unaffected.
\end{enumerate}

It is worth noting that this mechanism generalizes both the reservation based approach in which high-criticality tasks are allocated resources based on their high-confidence WCETs~\cite{LAA12}, as well as the classical mixed-criticality studies that penalize all the low-criticality tasks (e.g.,~\cite{Eas13}). Considering a mixed-criticality scheduling strategy based on the Earliest Deadline First (EDF) policy (e.g.,~\cite{BBA12,EkYi12,Eas13}), we also derive schedulability tests for the proposed mechanism. We derive these tests for a flat (non-hierarchical) as well as a hierarchical scheduling framework. While both these frameworks ensure isolation for high-criticality tasks as a result of employing criticality-aware scheduling, the hierarchical framework additionally supports \emph{compositionality}, i.e., the ability of a system to derive properties (e.g., schedulability) for higher level components using derived properties of lower level components. We evaluate the performance of the proposed mechanism in terms of schedulability and the ability to support low-criticality executions through extensive simulations. These results show that our proposed mechanism outperforms all the other existing studies in terms of this dual objective.

\textbf{Related Work.} Since Vestal's seminal work in 2007~\cite{Ves07}, a growing number of studies have been introduced for mixed-criticality real-time scheduling, e.g.,~\cite{BBD11,BaFo11,GES11,BBA12,EkYi12,Eas13}, 
sharing the pessimistic strategy that all the low-criticality tasks will be immediately dropped upon WCET violation of a single high-criticality task. Some recent studies have presented solutions to improve support for low-criticality executions~\cite{BuBa13,JZP13,Su13,PCH13,PCH14,fleming2014incorporating}. The elastic mixed-criticality model allows for a flexible release pattern of low-criticality tasks depending on the runtime resource consumption of high-criticality tasks, essentially treating the low-criticality workload as background~\cite{Su13,BuBa13,JZP13}. This was improved by the service adaptation strategy that decreased the dispatch frequency of low-criticality tasks only when a high-criticality task violated its low-confidence WCET. All the above studies however, share the unrealistic assumption that once a single high-criticality task violates its low-confidence WCET, all the other high-criticality tasks in the system will also exhibit similar behavior. The interference constraint graph strategy partially relaxes this assumption, at least in terms of its online strategy for penalizing low-criticality tasks~\cite{PCH13}. The constraint graph is used to specify execution dependencies between high- and low-criticality tasks, and a response-time based approach was presented to determine graph constraints that improve low-criticality executions at runtime. However, it still uses high-confidence WCET estimates for all the high-criticality tasks when determining schedulability (test based on~\cite{Ves07}), which again leads to the same unrealistic assumption and therefore results in resource under-utilization.  Further, none of the above studies consider the impact of WCET violations in the context of component-based systems. A couple of recent studies proposed techniques to support hierarchical scheduling for component-based mixed-criticality systems~\cite{HKM12,LAA12}. These studies focused on implementation issues however, and therefore did not consider the problems discussed above.

\section{System Model}
\label{sec:model_justification}
\subsection{Task and Component}
\label{sec:model}
In this paper we consider constrained deadline mixed criticality sporadic tasks (or \emph{tasks} for short). Such a task can be specified as $\tau_i = (T_i, L_i, \mathcal{C}_i, D_i)$, where $T_i$ denotes the minimum separation between job releases, $L_i$ denotes the criticality level, $\mathcal{C}_i$ is a list of WCET values, and $D_i$ ($\leq T_i$) denotes the relative deadline. We assume that tasks have only two criticality levels, $LC$ denoting low-criticality and $HC$ denoting high-criticality. Hence $L_i \in \{ LC, HC \}$ and $\mathcal{C}_i = \{ C_i^L, C_i^H \}$, where $C_i^L$ denotes $LC$ WCET and $C_i^H$ denotes $HC$ WCET. If $L_i = HC$, then $\tau_i$ is called a \emph{$HC$ task}, otherwise $\tau_i$ is called a \emph{$LC$ task}. We also assume that $C_i^L < C_i^H$ for all the $HC$ tasks, and $C_i^L = C_i^H$ for  all the $LC$ tasks.  Jobs of $\tau_i$ are released with a minimum separation of $T_i$ time units, and each job can execute for no more than $C_i^H$ time units ($C_i^L$ in the case of $LC$ task) within $D_i$ time units from its release. Let $\mathcal{T} = \{ \tau_1, \ldots , \tau_n \}$ denote a set of such mixed-criticality tasks that are scheduled on a single-core processor.

We assume that the tasks are partitioned into \emph{components}, where each component $\mathbb{C} = (\mathcal{W}, TL)$ comprises the following.
\begin{itemize}
\item A \emph{real time workload} $\mathcal{W}$ denoting a subset of tasks from $\mathcal{T}$, and
\item A \emph{$HC$ Tolerance Limit} $TL \in \mathbb{N}$ denoting the maximum $HC$ workload isolation limit of the component. As long as no more than $TL_i$ tasks in the component simultaneously exhibit $HC$ behavior (execution requirement is more than $LC$ WCET), we must ensure that all the job deadlines in the other components, including those of $LC$ jobs, are met. More details about this parameter are presented later in this section.
\end{itemize}

Partitioning the task set into components is mainly driven by practical considerations as mentioned in the introduction. Since these components are developed independently, it is desirable to limit the impact of WCET violations to within components as much as possible, while still efficiently utilizing the resources. The $HC$ tolerance limit $TL$ precisely does that in our model. It could be set based on component properties if information about the runtime behavior of $HC$ jobs is available, e.g., probability of execution requirement exceeding $LC$ WCET. It can also be determined such that the limit is maximized so as to support more $LC$ job executions, while still maintaining system schedulability. $\mathbb{C} = (\mathcal{W}, TL)$ is called a \emph{$LC$} component if every task in its workload is a $LC$ task, and for such components we assume that $TL = 0$. Otherwise, $\mathbb{C}$ is called a \emph{$HC$} component.

\subsection{Task and Component Execution Model}
\label{sec:execution_model}
The execution semantics of a mixed-criticality task has been presented previously~\cite{BBD11}, and we summarize it as follows. A task $\tau_i$ is said to be in low-criticality mode (or \emph{$LC$ mode} for short) as long as no job of the task has executed beyond its $LC$ WCET $C_i^L$. If $\tau_i$ is a $LC$ task, then this is the only available criticality mode. Whereas if $\tau_i$ is a $HC$ task, then it switches to high-criticality mode (or \emph{$HC$ mode} for short) at the time instant when some job of the task requests to execute for more than its $LC$ WCET. In $HC$ mode, jobs of $\tau_i$ can request to execute for no more than $C_i^H$ time units. 

We now define the execution semantics of a component $\mathbb{C} = (\mathcal{W}, TL)$. $\mathbb{C}$ has two execution modes, an \emph{internal mode} that concerns the behavior of tasks in $\mathbb{C}$, and an \emph{external mode} that concerns the behavior of tasks in the other components. We first describe these two modes, and then discuss their implications.

\textbf{Internal Mode.} Component $\mathbb{C}$ experiences \emph{Internal Mode Switch} (or IMS for short) at the earliest time instant when any $HC$ task in $\mathbb{C}$ switches to $HC$ mode. The component switches its internal mode from $LC$ to $HC$ at this time instant. Prior to this mode switch, all the task deadlines are required to be met. After this switch however, all the $LC$ tasks in $\mathbb{C}$ can be dropped, and only the $HC$ task deadlines are required to be met. There is no impact of this mode switch on the other components in the system. 

\textbf{External Mode.} Component $\mathbb{C}$ experiences \emph{External Mode Switch} (or EMS for short) at the earliest time instant when the $(TL+1)^{st}$ $HC$ task in $\mathbb{C}$ switches to $HC$ mode. The component switches its external mode from $LC$ to $HC$ at this time instant. Prior to this mode switch, at most $TL$ tasks in $\mathbb{C}$ were executing in $HC$ mode. After this switch however, all the $HC$ tasks in $\mathbb{C}$ may execute in $HC$ mode. Further, all the $LC$ tasks in the system, including the $LC$ tasks in the other components, are no longer required to meet deadlines. \ae{Component $\mathbb{C}$'s internal as well as external modes could switch back to $LC$ mode when there are no pending jobs in the  system at some time instant.}

Note that the intra- and inter-component execution requirements based on their internal and external modes respectively, are consistent with the mixed-criticality requirements in the existing literature (e.g,~\cite{BBA12}). If $\mathbb{C}$ is a $LC$ component, then its internal and external modes are identical and equal to $LC$. On the other hand if $\mathbb{C}$ is a $HC$ component, then these modes, together with the $HC$ tolerance limit $TL$, are key mechanisms for supporting $LC$ job executions. If $TL>0$, it is possible for IMS and EMS to occur at different time instants (asynchronously). Then, during the interval when component $\mathbb{C}$'s internal mode is $HC$ while its external mode is $LC$, $LC$ tasks in the other components are isolated from the internal mode switch of $\mathbb{C}$. That is, these $LC$ tasks can continue their execution even though some $HC$ tasks in $\mathbb{C}$ are already executing in $HC$ mode.

The proposed model and execution strategy generalizes both the worst-case reservation based approach in which $HC$ tasks are allocated resources based on their $HC$ WCETs~\cite{LAA12}, as well as the classical mixed-criticality studies that drop all the low-criticality tasks upon WCET violation by a single $HC$ task (e.g.,~\cite{Eas13}). The former can be modeled by setting $TL=|H|$, where $|H|$ denotes the total number of $HC$ tasks in the component, while the latter can be modeled by setting $TL=0$. 

\textbf{Scheduling Strategy.} In this paper we focus on the Earliest Deadline First (EDF) strategy, and assume that $LC$ tasks are dropped (not considered for scheduling) once it becomes known that their deadlines are not required to be met.
We have chosen this scheduling strategy because it has been successfully employed in the past for mixed criticality systems (e.g.,~\cite{BBA12,EkYi12,Eas13}). To accommodate the sudden increase in demand when tasks start executing in $HC$ mode, these existing studies artificially tighten the deadlines of $HC$ tasks when they are executing in $LC$ mode. This ensures that when a task switches to $HC$ mode, it has some amount of time left until its real deadline to execute any additional demand. In this paper we assume that such deadline tightening strategies are employed.

For a task $\tau_i = (T_i, L_i, \mathcal{C}_i, D_i)$, we let $D_i^L$ denote the artificially \emph{tightened deadline} in $LC$ mode of execution. By definition, $D_i^L \leq D_i$ for all tasks, and $D_i^L = D_i$ if $L_i=LC$ because no tightening is required for such $LC$ tasks. While a $HC$ task $\tau_i = (T_i, HC, \mathcal{C}_i, D_i)$ is executing in $LC$ mode, $\tau_i$ must receive at least $C_i^L$ processor units before its tightened deadline $D_i^L$. When the task $\tau_i$ switches to $HC$ mode, it must receive at least $C_i^H$ processor units before the actual deadline $D_i$. Note that a $HC$ task in component $\mathbb{C}$ that executes in $LC$ mode after IMS of $\mathbb{C}$ will continue to be scheduled using its tightened deadline $D_i^L$, unless it switches to $HC$ mode. After EMS of $\mathbb{C}$ however, all the $HC$ tasks are assumed to switch to $HC$ mode and will be scheduled using their actual deadlines.

We consider two different scheduling frameworks in this paper; a flat (non-hierarchical) framework in which all the tasks in all the components are collectively scheduled by a single scheduler, and a hierarchical framework in which the tasks in components are scheduled by intra-component schedulers and the components themselves are scheduled by a inter-component scheduler. The flat framework is relevant in applications that do not use hierarchical scheduling (e.g., Deos Real-Time Operating System for avionics~\cite{DO1781}), whereas the hierarchical framework is relevant in applications that require compositionality  (e.g., ARINC653 in avionics~\cite{ARINC653}). Note that a criticality-aware flat scheduler also ensures isolation for high-criticality tasks, and hence from that perspective provides similar functionality as a hierarchical scheduler. In Section~\ref{sec:dbf_test} we present the schedulability test under a flat scheduling framework, and in Section~\ref{sec:stest} we present the schedulability test under a hierarchical scheduling framework. Finally, the capability of the proposed mechanism and the corresponding schedulability tests to support $LC$ job executions are evaluated through extensive simulations in Section~\ref{sec:evaluation}.

\section{Schedulability Test for Flat Scheduling Framework}
\label{sec:dbf_test}

\emph{Demand bound function} ($\mbox{dbf}$), which gives an upper bound on the maximum possible execution demand of tasks in given time interval length, was first proposed to characterize the maximal demand of workloads comprising  non-mixed-criticality tasks~\cite{BMR90}. Since then $\mbox{dbf}$ has been extended to mixed-criticality tasks as well~\cite{EkYi12,Eas13}.

In this section, for the task and component model presented earlier, we propose a $\mbox{dbf}$-based schedulability test under an EDF-based flat scheduling framework. In Section~\ref{sec:job} we present the functions to calculate the demand of two special jobs of a task, and in Section~\ref{sec:task} we use this to compute the $\mbox{dbf}$ of a task (this $\mbox{dbf}$ has already been developed in~\cite{Eas13}). In Section~\ref{sec:dbfcomponent}, we present the $\mbox{dbf}$ of a component, and finally in Section~\ref{sec:flat} we present the $\mbox{dbf}$-based schedulability test. 

Let $t$ denote the time interval length and without loss of generality we assume the time interval is $[0,t)$. Let $t_E (\leq t)$ denote the time instant for External Mode Switch or EMS of $\mathbb{C}$, and $t_I (\leq t_E)$ denote the time instant for Internal Mode Switch or IMS of $\mathbb{C}$. If $\mathbb{C}$ is a $LC$ component, then it has no IMS or EMS, and tasks within it will be dropped after the earliest EMS of any component in the system. For a $HC$ task $\tau_i$ in the workload of $\mathbb{C}$, let $t_i$ denote the time instant when it switches to $HC$ mode. By definition $t_I \leq t_i \leq t_E$. For a $LC$ task $\tau_i$ in the workload of $\mathbb{C}$, let $t_i$ denote the time instant when it is dropped. Note that $t_i$ in the $LC$ case is either equal to $t_I$ or the earliest EMS of any $HC$ component, whichever is earlier.  We use $J_i$ to denote any job of $\tau_i$, and $r(J_i)$ to denote its release time.

\subsection{Demand of two special jobs}
\label{sec:job}
We now introduce how to compute the demand of the first special job which is the last one released by $HC$  task $\tau_i$ before it switches to $HC$ mode at $t_i$. As shown in  Figure~\ref{fig:carryoverd}, this is a job such that $r(J_i) \leq t_i$ and $r(J_i) + T_i > t_i$, and we denote such a job as $J_i^A$.
\begin{figure}[tbp]
\centering
\includegraphics[scale=0.65]{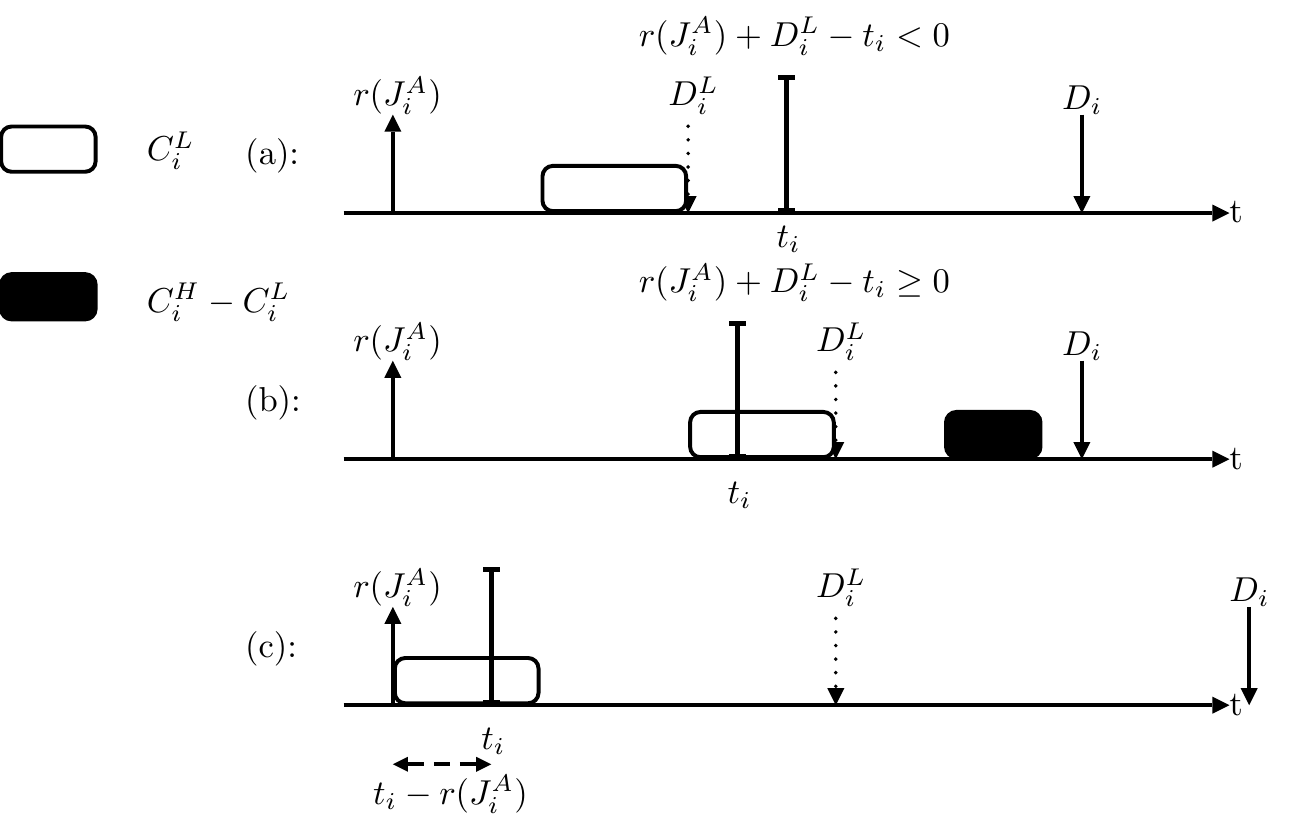}
\caption{Execution pattern for $J_i^A$ that generates maximal demand}
  \label{fig:carryoverd}
\end{figure}
The following lemma bounds the demand of $J_i^A$ when its deadline is greater than $t$. 
\begin{lemma}
\label{lemma:2}
If $r(J_i^A)+D_i^L > t$, then $J_i^A$ will generate zero demand during $[0,t)$. 
Further, if $r(J_i^A)+D_i > t$, then $J_i^A$ will not generate any demand after  $t_i$.
\end{lemma}
\begin{proof}
Since $r(J_i^A)+D_i^L>t\Rightarrow r(J_i^A)+D_i>t~(D_i\geq D_i^L)$, $J_i^A$ does not generate any demand in the interval of interest. On the other hand, if $r(J_i^A)+D_i > t$ and $r(J_i^A)+D_i^L\leq t$, then even if $J_i^A$ does not finish before $t_i$, it does not generate any demand in the interval $[t_i, t)$ because after $t_i$ its deadline is outside the interval of interest. 
\end{proof}

If $J_i^A$ satisfies the condition $r(J_i^A)+D_i^L < t_i$ (Figure~\ref{fig:carryoverd}(a)), then $J_i^A$ must finish by $t_i$, and hence it can generate a demand of up to $C_i^L$ during $[0,t)$. However if $r(J_i^A)+D_i^L\geq t_i$ (Figure~\ref{fig:carryoverd}(b)), then $J_i^A$ will generate maximal demand during $[0,t)$ if it executes as late as possible. In this case it can generate a demand of up to $C_i^H$. One special case is when $t_i\leq r(J_i^A)+D_i^L\leq t$ and $r(J_i^A)+D_i>t$ (Figure~\ref{fig:carryoverd}(c)). In this case, $J_i^A$ will not generate any demand after $t_i$ according to Lemma~\ref{lemma:2}. Thus, the demand of job $J_i^A$ for the interval $[0,t)$ can be bounded as follows.
\begin{equation}
\label{eqn:carryoverjob}
\begin{split}
\mbox{dbf}(J_i^A\!\!,t,t_i)\!=\!\!\!
\begin{cases}
C_i^L,&r(J_i^A)+D_i^L < t_i\\ 
C_i^H,&r(J_i^A)+D_i^L \geq t_i \!\!\!\\&\mbox{and } r(J_i^A)+D_i\leq t  \\
\min\left\{\!t_i\!-\!r(J_i^A),\!C_i^L\! \right\},&t_i\leq r(J_i^A)+D_i^L\leq t\\
                     &\mbox{and } (J_i^A)+D_i>t\\
0,&r(J_i^A)+D_i^L>t
\end{cases}
\end{split}
\end{equation}

Another special job   is the last job released by a $LC$ task $\tau_i$ before it is dropped at $t_i$, and we denote such a job as $J_i^B$. The release time of $J_i^B$ satisfies the conditions $r(J_i^B)\leq t_i$ and $r(J_i^B)+T_i>t_i$.

If $r(J_i^B)+D_i^L>t$ $(D_i^L=D_i)$, $J_i^B$ will generate zero demand during $[0,t)$ because its deadline is outside the interval. Otherwise, it may generate some demand in the interval $[0, t_i)$, because it will be dropped after $t_i$. In order to maximize the demand of $J_i^B$ in this interval, we assume that $J_i^B$ will execute continuously from $r(J_i^B)$. Thus, the demand of job $J_i^B$ for the interval $[0,t)$ can be bounded as follows.
\begin{equation} 
\label{eqn:unnecessary}
\begin{split}
\mbox{dbf}(J_i^B,t,t_i)=
\begin{cases}
\min\left\{t_i-r(J_i^B),C_i^L \right\},& r(J_i^B)+D_i^L \leq t\\
0,~~&\mbox{otherwise}
\end{cases}
\end{split}
\end{equation}

\subsection{$\mbox{Dbf}$ of  task $\tau_i$}
\label{sec:task}
In this section we derive the $\mbox{dbf}$ of a task $\tau_i$ using Equations~\ref{eqn:carryoverjob} and \ref{eqn:unnecessary} presented above. Let $\mbox{dbf}(\tau_i,t,t_i)$ denote the $\mbox{dbf}$ of task $\tau_i$ for a given time interval length $t$ and instant $t_i$.  We present $\mbox{dbf}(\tau_i,t,t_i)$ using four sub-cases $\mbox{dbf}(\tau_i,t,t_i)_{[x]}$, where $x \in\left\{a,b,c,d\right\}$,  defined as follows.
\begin{description}
\item[\textbf{a}: ] $L_i=LC$,
\item[\textbf{b}: ] $L_i=HC$ and $t-t_i< D_i-D_i^L$,
\item[\textbf{c}: ] $L_i=HC$ and  $t-t_i\geq D_i$, and
\item[\textbf{d}: ]$L_i=HC$ and  $D_i-D_i^L\leq t-t_i<D_i$.
\end{description}

If $\tau_i$ satisfies condition a, then it is a $LC$ task. The total demand that $\tau_i$ can generate during $[0,t)$ is then the sum of demand of jobs released before $r(J_i^B)$ and the demand of $J_i^B$ itself. $\tau_i$ generates maximal demand during $[0,t)$ if the release time of the first job is equal to zero, and all successive jobs are released as soon as possible with period $T_i$. Therefore $\mbox{dbf}(\tau_i,t,t_i)_{[a]}$ is given as follows.

\begin{equation}
\mbox{dbf}(\tau_i,t,t_i)_{[a]}=\left \lfloor \frac{t_i}{T_i} \right \rfloor C_i^L+\mbox{dbf}(J_i^B,t,t_i)
\end{equation}
If $\tau_i$ satisfies condition b, c or d,  then $\tau_i$ is a $HC$ task. Therefore, the total demand it generates is the sum of demand of all the jobs released before $t$. Among these jobs, the ones released before $r(J_i^A)$ will generate a demand of $C_i^L$, and the ones released after $r(J_i^A)+T_i$ will generate a demand of $C_i^H$. The demand of job $J_i^A$ itself is given in Equation~\ref{eqn:carryoverjob}. In the following lemmas we derive the $\mbox{dbf}$ of $\tau_i$ for the three conditions.
\begin{figure}
\includegraphics[scale=0.65]{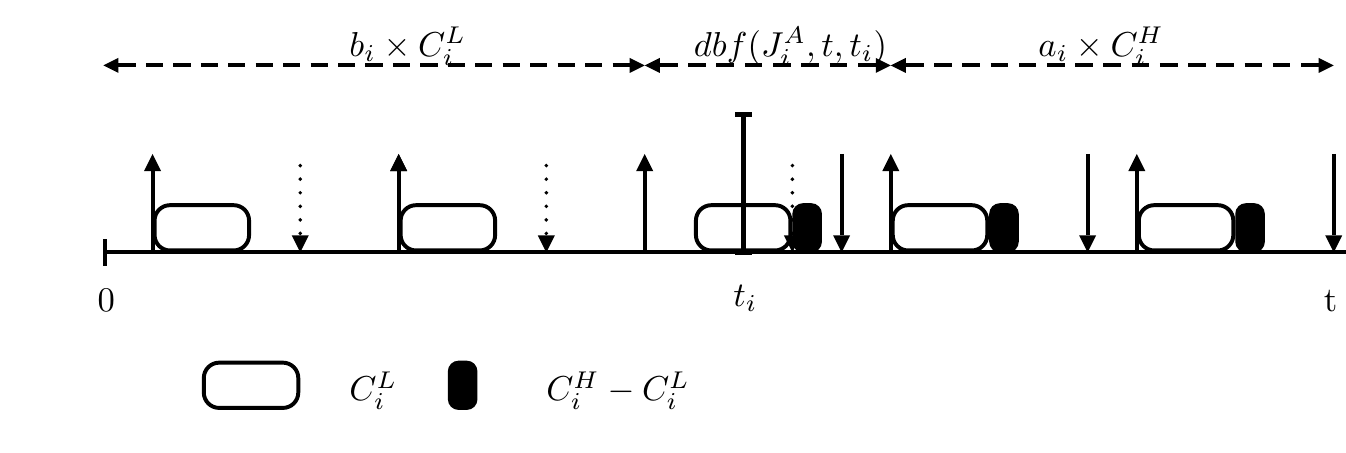} 
\caption{Execution pattern for condition c}
\label{fig:HC}
\end{figure}
\begin{lemma}
\label{lemma:HCcase1}
If $\tau_i$ satisfies condition b ($t-t_i< D_i-D_i^L$), no job of $\tau_i$ can execute for $C_i^H$ time units. Therefore $\tau_i$ can generate maximal demand during $[0,t)$ if the first job of $\tau_i$ is released at time instant $0$ and all the successive jobs are released as soon as possible.
\end{lemma}
\begin{proof}
We prove this lemma by contradiction. Suppose there exists a job $J_i$ of $\tau_i$ that can generate a demand of $C_i^H$ time units in the interval $[0, t)$. Then it must be true that $r(J_i)+D_i \leq t$ and $r(J_i)+D_i^L \geq t_i\Rightarrow t-t_i\geq D_i-D_i^L$, because $\tau_i$ is a $HC$ task that switched to $HC$ mode at $t_i$. This contradicts our assumption that $t-t_i<D_i-D_i^L$. Thus no job of $\tau_i$ that satisfies condition b can generate a demand of $C_i^H$ time units in the interval $[0,t)$. Therefore $\tau_i$ essentially behaves like a $LC$ task, and this proves the lemma.
\end{proof}
Thus $\mbox{dbf}(\tau_i,t,t_i)_{[b]}$ is given as follows.
\begin{equation}
\setlength{\abovedisplayskip}{3pt}
\setlength{\belowdisplayskip}{3pt}
\mbox{dbf}(\tau_i,t,t_i)_{[b]}=\left \lfloor \frac{t_i}{T_i} \right \rfloor C_i^L+\mbox{dbf}(J_i^A,t,t_i) 
\end{equation}

\begin{lemma}
\label{lemma:HCcase2}
If $\tau_i$ satisfies condition c ($t-t_i\geq D_i$), it generates maximal demand during $[0,t)$ if the first job of $\tau_i$ is released at $t-D_i-\lfloor {(t-D_i)}/{T_i} \rfloor \times T_i$, and all the successive jobs are released as soon as possible (scenario shown in Figure~\ref{fig:HC}).
\end{lemma}
\begin{proof}
If $t-t_i \geq D_i$ and the first job of $\tau_i$ is released at $t-D_i-\lfloor {(t-D_i)}/{T_i} \rfloor \times T_i$, then the last job released before $t$ will have its deadline at $t$. In this case, $t_i$ happens before the release time of this last job. Therefore the last job can generate a demand of $C_i^H$ in the interval. Additionally, the number of jobs with deadline before $t$ as well as the number of jobs that can generate $C_i^H$ demand during $[0,t)$ are maximized with this pattern. This proves the lemma.
\end{proof}
Intuitively speaking, the demand is maximized when the deadline of a job of $\tau_i$ coincides with $t$, because it maximizes the possible executions for $\tau_i$ in $HC$ mode. Thus, $\mbox{dbf}(\tau_i,t,t_i)_{[c]}$ is given as follows.
\begin{equation}
\begin{split}
\mbox{dbf}(\tau_i,t,t_i)_{[c]}&=b_i C_i^L+\mbox{dbf}(J_i^A,t,t_i)+a_i  C_i^H~\mbox{, where}\\
b_i &= \left \lfloor \frac{t_i-(t-D_i-\lfloor {(t-D_i)}/{T_i}\rfloor\times T_i)}{T_i} \right\rfloor, \mbox{and} \\
a_i  &= \left\lfloor \frac{t-D_i}{T_i}\right \rfloor-b_i.
\label{eqn:condc}
\end{split}
\end{equation}

If $\tau_i$ satisfies condition d, it does not have a single execution pattern that maximizes its demand as stated in the following lemma.   
\begin{lemma}
\label{lemma:HCcase3}
If $\tau_i$ satisfies condition d ($D_i-D_i^L\leq t-t_i<D_i$), it generates maximal demand if its first job is  either released at $0$ (condition b) or at $t-D_i-\lfloor {(t-D_i)}/{T_i}\rfloor\times T_i$ (condition c).
\end{lemma}
\begin{proof}
Since $D_i-D_i^L\leq t-t_i< D_i$, $\tau_i$ can have at most one job that can generate a demand of $C_i^H$ in the interval $[0,t)$. If the first job of $\tau_i$ is released at $t-D_i-\lfloor (t-D_i)/T_i \rfloor T_i$ and all the successive jobs are released as soon as possible (release pattern of condition c), then the last job is a special job $J_i^A$ and is the only job generating $C_i^H$ demand. The only way to further increase the demand of $\tau_i$ is to add a new job in the interval by shifting the pattern left to the point when  the first job is released at time instant $0$. 
\end{proof}
Thus, $\mbox{dbf}(\tau_i,t,t_i)_{[d]}$ is given as follows.
\begin{equation}
\mbox{dbf}(\tau_i,t,t_i)_{[d]}=\max\left\{\mbox{dbf}(\tau_i,t,t_i)_{[b]},\mbox{dbf}(\tau_i,t,t_i)_{[c]} \right\}
\end{equation}

\subsection{$\mbox{Dbf}$ of component $\mathbb{C}$}
\label{sec:dbfcomponent}
In this section we  present the $\mbox{dbf}$  of a component $\mathbb{C} = \{\mathcal{W}, TL \}$. Let $\mbox{dbf}(\mathbb{C},t,t_E,t_I)$  denote the $\mbox{dbf}$ of component $\mathbb{C}$ for a given time interval length $t$, with mode-switch instants $t_I$ (IMS) and  $t_E$ (EMS).

We first present $\mbox{dbf}$ for the case when $TL=0$ and then for  the case when $TL > 0$. Note that among all the $HC$ tasks in $\mathbb{C}$, at most $TL$ of them can switch to $HC$ mode in the interval $[t_I, t_E)$, while all the remaining $HC$ tasks are assumed to switch to $HC$ mode at $t_E$.

If $TL=0$, then this means $t_E$ (EMS) is equal to $t_I$ (IMS), because $\mathbb{C}$'s internal and external modes will switch at the same time. Thus, each $HC$ task $\tau_i$ in $\mathbb{C}$ will switch to $HC$ mode at $t_i = t_I =t_E$, and hence $\mbox{dbf}(\mathbb{C},t,t_E=t_I,t_I)$ is given as follows.
\begin{equation}
\setlength{\abovecaptionskip}{-0.35cm} 
\setlength{\belowcaptionskip}{-0.35cm} 
\mbox{dbf}(\mathbb{C},t,t_E=t_I,t_I)=\sum_{\tau_i\in \mathbb{C}}\mbox{dbf}(\tau_i,t,t_I) ~~~(TL=0)
\end{equation}
If $TL > 0$, then at most $TL$ $HC$ tasks can switch to $HC$ mode before $t_E$. To compute the $\mbox{dbf}$ of $\mathbb{C}$, we then need to determine which $HC$ tasks should switch to $HC$ mode before $t_E$ so as to maximize the total demand. The following lemma asserts that for any $HC$ task, its demand is maximized when it switches to $HC$ mode either at $t_I$ or $t_E$.

\begin{lemma}
\label{lemma:max}
If  a $HC$ task $\tau_i$ switches to $HC$ mode at some time $t_i \in [t_I,t_E]$, then $\mbox{dbf}(\tau_i,t,t_i)$ is maximized when $t_i$ is either equal to $t_E$ or $t_I$.
\end{lemma}
\begin{proof}
\ae{
Suppose $\tau_i$ satisfies condition b when $t_i = t_E$, i.e., $t-t_E <D_i - D_i^L$. Then as $t_i$ decreases, $\tau_i$ could eventually satisfy condition d, i.e., $D_i - D_i^L \leq t-t_i < D_i$, and finally condition c, i.e., $t-t_i\geq D_i$. Without loss of generality, assume that $\tau_i$ satisfies condition b for $t_i \in (t_b, t_E]$, condition d for $t_i \in (t_d, t_b]$, and condition c for $t_i \in [t_I, t_d]$, where $t_I\leq t_d\leq t_b \leq t_E$. 

\textbf{Case 1}($t_i\in [t_I, t_d]$): In this case, $\mbox{dbf}(\tau_i,t,t_i)_{[c]}$ (see Equation~\ref{eqn:condc}) is maximized if $t_i=t_I$. This is because as $t_i$ decreases from $t_d$ to $t_I$, the number of jobs generating $C_i^H$ demand will remain the same or increase, while the total number of jobs that generate demand for this time interval remains unchanged.  \textbf{Case 2} ($t_i\in (t_b, t_E]$): In this case, $\mbox{dbf}(\tau_i,t,t_i)_{[b]}=\left \lfloor \frac{t_i}{T_i} \right \rfloor C_i^L+\mbox{dbf}(J_i^A,t,t_i)$. Then as $t_i$ increases from $t_b$ to $t_E$, $\mbox{dbf}(J_i^A,t,t_i)$ and $\left \lfloor \frac{t_i}{T_i} \right \rfloor\times C_i^L$ will stay the same or increase. Thus $\mbox{dbf}(\tau_i,t,t_i)_{[b]}$ is maximized when $t_i=t_E$.  \textbf{Case 3} ($t_i\in(t_d, t_b]$): From Lemma~\ref{lemma:HCcase3} we know that $\mbox{dbf}(\tau_i,t,t_i)_{d}=\max\left\{\mbox{dbf}(\tau_i,t,t_i)_{[b]},\mbox{dbf}(\tau_i,t,t_i)_{[c]}\right\}$. While  $\mbox{dbf}(\tau_i,t,t_i)_{[b]}|t_i\in(t_d, t_b]$ is maximized if $t_i=t_b$, $\mbox{dbf}(\tau_i,t,t_i)_{[c]}|t_i\in(t_d, t_b]$ stays the same.  Since $\mbox{dbf}(\tau_i,t,t_b)_{[b]}\leq \mbox{dbf}(\tau_i,t,t_E)_{[b]}$ and $\mbox{dbf}(\tau_i,t,t_b)_{[c]}\leq \mbox{dbf}(\tau_i,t,t_I)_{[c]}$, combining the above three cases, we conclude that  $\mbox{dbf}(\tau_i,t,t_i)$ is maximized when $t_i=t_I$ or $t_i=t_E$. 
}
\end{proof}
\ae{
Let $\Delta_i=\max\{0,\mbox{dbf}(\tau_i,t,t_I)-\mbox{dbf}(\tau_i,t,t_E)\}$. From Lemma~\ref{lemma:max} we know that task $\tau_i$ generates maximum demand when $t_i=t_E$ or $t_i=t_I$.  Therefore $\Delta_i$ denotes the maximum possible increase in the demand of $\tau_i$ (if it increases) for a time interval length $t$ when $\tau_i$ is chosen as one of the $TL$ tasks to switch to $HC$ mode before $t_E$. Once we compute $\Delta_i$ for all the $HC$ tasks in component $\mathbb{C}$, we sort these values in descending order and select the first $TL$ elements. Let the corresponding set of $TL$ $HC$ tasks be denoted by $\mathcal{G}$. The total maximum demand of all the tasks in $\mathbb{C}$ is then given by the following equation.
}
\begin{equation}
\setlength\abovedisplayskip{1pt} 
\setlength\belowdisplayskip{1pt} 
\begin{split}
\mbox{dbf}(\mathbb{C},t,t_E,t_I)=&\sum\limits_{L_i=HC}\mbox{dbf}(\tau_i,t,t_E)+ \sum\limits_{\tau_i\in \mathcal{G}} \Delta_i\\&+\sum\limits_{L_i=LC}\mbox{dbf}(\tau_i,t,t_I)
\end{split}
\end{equation}
A tighter bound for the $\mbox{dbf}$ of component $\mathbb{C}$ can be obtained using an optimization presented in Section~\ref{sec:opt} of the Appendix.
\subsection{Schedulability Test and Tolerance Limit}
\label{sec:flat}
In this section we derive the schedulability test for a mixed-criticality system comprising  multiple components and scheduled under a flat scheduling framework. Consider a system  with $p$ $HC$ components $\mathbb{C}_{1},\mathbb{C}_{2},\ldots,\mathbb{C}_{p}$ and $q$ $LC$ components $\mathbb{C}_{p+1},\mathbb{C}_{p+2},\ldots,\mathbb{C}_{p+q}$. Each $HC$ component $\mathbb{C}_{i}$ can independently switch its internal mode to $HC$ at $t_{Ii}$. Once the first $HC$ component switches its external mode to $HC$ at $t_E$, all the $LC$ tasks in the system are immediately dropped. We assume that all the $HC$ tasks in the system can thereafter execute in $HC$ mode.

Suppose there is a first deadline miss in the system at some time instant $t$. Then, the total maximum demand generated by the system in $[0,t)$ must be greater than $t$. This assertion immediately leads to the following theorem that presents the schedulability test.

\begin{theorem}
\label{theorem:test1}
A mixed-criticality system comprising  $p$ $HC$ components and $q$ $LC$ components is schedulable under a flat scheduling framework if,
$\forall t: 0 \leq t \leq t_{MAX}, \forall t_E: 0 \leq t_E \leq t, \forall t_{Ii}: 0 \leq t_{Ii}\leq t_E$, 
\\
\begin{equation}
\sum\limits_{i=1}^{i\leq p+q} \mbox{dbf}(\mathbb{C}_{i},t,t_E,t_{Ii})\leq t,
\end{equation}
where $t_{MAX}$ is a pseudo-polynomial in the size of the input, and is defined in Section~\ref{sec:tmax} of the Appendix.
\label{thm:thm_hc1}
\end{theorem}
The complexity of the schedulability test in Theorem~\ref{thm:thm_hc1} is exponential in the number of $HC$ components, because we need to consider a separate internal mode switch instant for each component. In practice however, we expect the number of $HC$ components scheduled on a single processor to be relatively small, and then the complexity of the proposed test is pseudo-polynomial in the size of the input. Besides, if there is freedom to select the allocation of system tasks to components, then it is feasible to create a component structure comprising only two components, while still fully supporting $LC$ task executions. All the $HC$ tasks in the system are allocated to a single $HC$ component $\mathbb{C}_H=\{\mathcal{W}_H,TL_H\}$, and each $LC$ task can be either allocated to $\mathbb{C}_H$ or to a $LC$ component $\mathbb{C}_L=\{\mathcal{W}_L,TL_L=0\}$. This two-component system is sufficient to consider all the possible design choices for isolating $HC$ and $LC$ task executions. This can be done by considering different values for the tolerance limit $TL_H$, and by considering different allocations of $LC$ tasks to the two components. We can choose the maximum possible value for these tolerance limit as long as the resulting system is still schedulable. Higher tolerance limit indicates support for more $LC$ task executions, and thus better resource utilization. In Section~\ref{sec:evaluation}, we show through simulations that our mechanism outperforms existing studies even with this two-component structure. However, if the allocation of tasks to components is fixed and the number of $HC$ components is not small, then the hierarchical scheduling framework presented in the following section can be used to reduce the complexity of the test.  


\section{Schedulability Test for Hierarchical Scheduling Framework}
\label{sec:stest}
Hierarchical scheduling has emerged as an effective mechanism to support temporal partitioning between applications, serving as a common scheduling paradigm in many mixed-criticality systems in practice~\cite{ARINC653}. It is preferred in practice because it supports \emph{compositionality} so that higher-level properties can be derived from verified component-level properties. Therefore, to increase the practical relevance of the proposed mechanism, we develop a schedulability test under a hierarchical scheduling framework in this section.

\subsection{Execution Strategy under Hierarchical Scheduling}
\label{sec:component_interface}
For hierarchical systems, each component will have an additional parameter $S$ denoting its local scheduler. We specify such a component as $\mathbb{C} = (\mathcal{W}, TL, S)$. The component workload $\mathcal{W}$ is comprised of regular mixed-criticality tasks as well as interface tasks representing the child components. The tasks in the workload $\mathcal{W}$ are scheduled by the local scheduler $S$, independently of all the other components in the system.

\emph{Component interfaces} have been widely used in traditional hierarchical systems to abstractly represent the resource demand and supply of components (see for example~\cite{ShLe03}). From the component's perspective, its interface represents the resource demand of its workload. While from the perspective of its parent component or system, the interface represents the resource supply that the parent guarantees. These interfaces of components are essential for satisfying the property of compositionality.

Resource models such as periodic have been previously defined as interfaces for components in traditional hierarchical systems~\cite{ShLe03}. Analogously, we now present the mixed-criticality periodic resource (MCPR) model  for mixed-criticality components. Since we focus on systems with two criticality levels, we assume that the MCPR model can have at most two criticality levels.
\begin{definition}
A Mixed-Criticality Periodic Resource (MCPR) is defined as $\mathbb{I} = (T, L, \mathcal{C})$, where $T$ denotes the period, $L\in \{ LC, HC \}$ denotes the criticality level, and $\mathcal{C}=\{ C^L, C^H \}$ is a list of resource capacities.  $C^L$ denotes $LC$ resource capacity and $C^H$ denotes $HC$ resource capacity. 
\label{def:mcpr}
\end{definition}

A component $\mathbb{C}$ can be abstracted as an MCPR interface $\mathbb{I} = (T, L, \mathcal{C})$, and the corresponding task $(T, L, \mathcal{C}, T)$ (denoted as \emph{interface task}) represents $\mathbb{C}$ in the workload of its parent component. We assume that period $T$ of this interface is already specified by the system designer as in the standard literature on hierarchical scheduling (e.g., see~\cite{ShLe03}). For instance, this period could be determined based on either component-level requirements or considerations for overheads such as context-switches. The criticality level $L$ is directly determined by the criticality level of the component it is representing. If $\mathbb{C}$ is a $LC$ component, then $L = LC$, otherwise $L = HC$.

\textbf{Mode of the interface.} The semantics of interface $\mathbb{I}$ (and the corresponding interface task) depend on its \emph{criticality mode} at run time, which in turn depends on the criticality mode of  component $\mathbb{C}$. In fact, we assume that the criticality mode of $\mathbb{I}$ is identical to the external mode of $\mathbb{C}$. When $\mathbb{C}$ experiences EMS, the mode of the interface and interface task switches from $LC$ to $HC$. While the interface is in $LC$ mode, it is guaranteed to request no more than $C^L$ time units of resource periodically every $T$ time units from the parent component. But when it switches to $HC$ mode, it can thereafter request up to $C^H$ time units of resource periodically.

\subsection{MCPR Supply Bound Function}
\label{sec:mcpr_sbf}
The supply bound function ($\mbox{sbf}$) of a resource model characterizes the minimum resource supply guaranteed by the model to the underlying component.  In this section, we derive the $\mbox{sbf}$ for a MCPR interface $\mathbb{I} = (T, L, \mathcal{C})$ of a component $\mathbb{C} = (\mathcal{W}, TL, S)$. We let $\mbox{sbf}_{\mathbb{I}} (t_E, t)$ denote the $\mbox{sbf}$ for a time interval of length $t$, where $t_E (\leq t)$ denotes the time instant for EMS of component $\mathbb{C}$. As the resource is supplied periodically, component $\mathbb{C}$ is guaranteed to receive either $C^L$ or $C^H$ units of resource every $T$ time units in $LC$ or $HC$ mode, respectively. We use the following additional notations in this section.

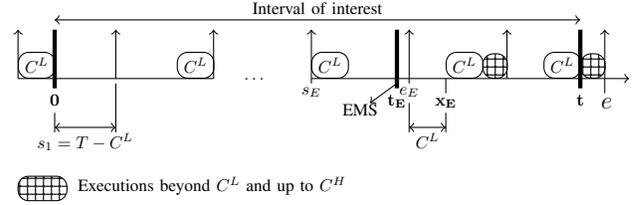
\begin{figure}[tbp]
\begin{tikzpicture}[scale=0.65,transform shape]
\coordinate (A) at (0,0);
\coordinate (A1) at (0,1);
\coordinate (B) at (2,0);
\coordinate (B1) at (2,1);
\coordinate (C) at (4,0);
\coordinate (C1) at (4,1);
\coordinate (D) at (6,0);
\coordinate (D1) at (6,1);
\coordinate (E) at (8,0);
\coordinate (E1) at (8,1);
\coordinate (F) at (10,0);
\coordinate (F1) at (10,1);
\coordinate (F2) at (8.75,0);
\coordinate (G) at (12,0);
\coordinate (G1) at (12,1);
\coordinate (K) at (0.75,0);
\coordinate (L) at (3.25,0);
\coordinate (M) at (6,0);
\coordinate (N) at (8.75,0);
\coordinate (N1) at (9.5,0);
\coordinate (O) at (10.75,0);
\coordinate (O1) at (11.5,0);
\coordinate (R) at (7.75,0);
\coordinate (S) at (11.5,0);
\draw (A) -- (C);
\draw [->] (D) -- (12.5,0);
\node [draw=none] [right =0.5cm of C]{$\ldots$};
\draw [->] (A) -- (A1);
\draw [->] (B) -- (B1);
\draw [->] (C) -- (C1);
\draw [->] (D) -- (D1);
\draw [->] (E) -- (E1);
\draw [->] (F) -- (F1);
\draw [->] (G) -- (G1);
\node [rounded corners,draw=black,minimum width=0.75cm] [above right =0cm of A]{$C^L$};
\node [rounded corners,draw=black,minimum width=0.75cm] [above right =0cm of L]{$C^L$};
\node [rounded corners,draw=black,minimum width=0.75cm] [above right =0cm of M]{$C^L$};
\node [rounded corners,draw=black,minimum width=0.75cm] [above right =0cm of N]{$C^L$};
\node [rounded corners,draw=black,pattern=grid,minimum height=0.5cm,minimum width=0.5cm] [above right =0cm of N1]{};
\node [rounded corners,draw=black,minimum width=0.75cm] [above right =0cm of O]{$C^L$};
\node [rounded corners,draw=black,pattern=grid,minimum height=0.5cm,minimum width=0.5cm] [above right =0cm of O1]{};
\coordinate (W) at (7,0);
\draw [->] (7.75,-0.1) -- (7.2,-0.5);
\node[draw=none] [below =0.4cm of W]{EMS};
\node[draw=none] [below =0.2cm of K]{$\mathbf{0}$};
\node[draw=none] [below =0.2cm of R]{$\mathbf{t_E}$};
\node[draw=none] [below =0.27cm of F2]{$\mathbf{x_E}$};
\node[draw=none] [below =0.2cm of S]{$\mathbf{t}$};
\node[draw=none] [below =0.05cm of D]{$s_E$};
\node[draw=none] [below =0.05cm of E]{$e_E$};
\node[draw=none] [below right=0.27cm and -0.2cm of G]{{\Large $e$}};

\begin{scope}[on background layer]
\draw [ultra thick] (0.75,-0.2) -- (0.75,1);
\draw [ultra thick] (7.75,-0.2) -- (7.75,1);
\draw [ultra thick] (11.5,-0.2) -- (11.5,1);
\end{scope}
\draw [<->] (0.75,1.2) to node [above]{Interval of interest} (11.5,1.2);
\draw [<->] (0.75,-1) to node [below]{$s_1=T-C^L$} (2,-1);
\draw [<->] (8,-1) to node [below]{$C^L$} (8.75,-1);
\draw (0.75,-0.7) -- (0.75,-1.1);
\draw (2,-0.1) -- (2,-1.1);
\draw (6,0) -- (6,-0.1);
\draw (8,-0.5) -- (8,-1.1);
\draw (8,0) -- (8,-0.1);
\draw (8.75,-0.7) -- (8.75,-1.1);
\draw (8.75,0) -- (8.75,-0.3);
\draw (12,0) -- (12,-0.3);
\coordinate (x) at (0,-2);
\coordinate (x1) at (1,-2.5);
\draw [rounded corners, pattern=grid] (x) rectangle (x1);
\node[draw=none] [above right=0cm and 0.1cm of x1]{Executions beyond $C^L$ and up to $C^H$};
\end{tikzpicture}
\caption{MCPR worst-case resource supply pattern~A}
\label{fig:sbf_pattern_2}
\end{figure}
\begin{figure}[tbp]
\begin{tikzpicture}[scale=0.65,transform shape]
\coordinate (A) at (0,0);
\coordinate (A1) at (0,1);
\coordinate (B) at (2,0);
\coordinate (B1) at (2,1);
\coordinate (C) at (4,0);
\coordinate (C1) at (4,1);
\coordinate (D) at (6,0);
\coordinate (D1) at (6,1);
\coordinate (E) at (8,0);
\coordinate (E1) at (8,1);
\coordinate (F) at (10,0);
\coordinate (F1) at (10,1);
\coordinate (F2) at (12.75,0);
\coordinate (G) at (12,0);
\coordinate (G1) at (12,1);
\coordinate (K) at (0.75,0);
\coordinate (L) at (3.25,0);
\coordinate (M) at (7.25,0);
\coordinate (N) at (9.25,0);
\coordinate (O) at (10.75,0);
\coordinate (O1) at (11.5,0);
\coordinate (R) at (10.9,0);
\coordinate (S) at (11.4,0);
\draw (A) -- (C);
\draw [->] (D) -- (12.5,0);
\node [draw=none] [right =0.5cm of C]{$\ldots$};
\draw [->] (A) -- (A1);
\draw [->] (B) -- (B1);
\draw [->] (C) -- (C1);
\draw [->] (D) -- (D1);
\draw [->] (E) -- (E1);
\draw [->] (F) -- (F1);
\draw [->] (G) -- (G1);
\node [rounded corners,draw=black,minimum width=0.75cm] [above right =0cm of A]{$C^L$};
\node [rounded corners,draw=black,minimum width=0.75cm] [above right =0cm of L]{$C^L$};
\node [rounded corners,draw=black,minimum width=0.75cm] [above right =0cm of M]{$C^L$};
\node [rounded corners,draw=black,minimum width=0.75cm] [above right =0cm of N]{$C^L$};
\node [rounded corners,draw=black,minimum width=0.75cm] [above right =0cm of O]{$C^L$};
\node [rounded corners,draw=black,pattern=grid,minimum height=0.5cm,minimum width=0.5cm] [above right =0cm of O1]{};
\node[draw=none] [below =0.5cm of K]{$\mathbf{0}$};
\node[draw=none] [below =0.5cm of R]{$\mathbf{t_E}$};
\node[draw=none] [below =0.5cm of S]{$\mathbf{t}$};
\node[draw=none] [below =0.05cm of F]{$s_E$};
\node[draw=none] [below right=0.05cm and -0.2cm of G]{$e$};
\node[draw=none] [below right=0.3cm and -0.35cm of G]{$(e_E)$};
\begin{scope}[on background layer]
\draw [ultra thick] (0.75,-0.5) -- (0.75,1);
\draw [ultra thick] (10.9,-0.5) -- (10.9,1);
\draw [ultra thick] (11.4,-0.5) -- (11.4,1);
\end{scope}
\draw [<->] (0.75,1.2) to node [above]{Interval of interest} (11.4,1.2);
\draw (10,0) -- (10,-0.1);
\draw (12,0) -- (12,-0.1);
\end{tikzpicture}
\caption{Boundary case for MCPR worst-case resource supply pattern A}
\label{fig:sbf_pattern_boundary}
\end{figure}
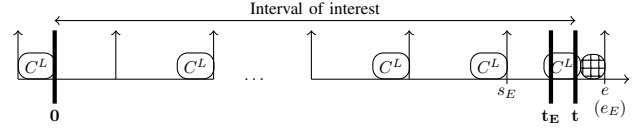

\begin{itemize}
\item $s_1$  denotes the start time of the first interface period within time interval $[0,t)$.
\item $n$ denotes the number of interface periods within interval $[0, t)$.
\item  $n_E$ denotes the number of interface periods within interval $[0, t_E)$.
\item $s_E$ denotes the start of a interface period that experiences EMS ($t_E$), i.e., $s_E\leq t_E<e_E$, where $e_E=s_E+T$.
\item $e$ denotes the start of interface period after $t$, i.e., $e=s_1+n \times T+T$.
\item For simplicity of presentation, we also use the short-cut notation $[x]_0 =\max\{ 0, x \}$.
\end{itemize}

When $t_E = t$, there is no external mode switch for component $\mathbb{C}$ in the interval of interest, and the component and interface are only executing in $LC$  mode. Therefore, $\mbox{sbf}_{\mathbb{I}} (t_E, t)$ in this case is identical to the $\mbox{sbf}$ defined for periodic resource models with $\mathbb{I}$ supplying $C^L$ units of resource periodically~\cite{ShLe03}. Thus, in this case, minimal resource is supplied when $s_1=T-C^L$ and $n = \left [ \left \lfloor \frac{t-(T-C^L)}{T} \right \rfloor \right ]_0$. We record this $\mbox{sbf}$ in the following equation.
\begin{equation}
\mbox{sbf}_{\mathbb{I}}(t_E,~t)\! =\! n\!  \times \! C^L\! + \! \left[t\! -\! 2(T-C^L)\! - n\!  \times T \right]_0~\mbox{ If }t_E=t
\label{eqn:sbf_pattern_1}
\end{equation}
For the case when $t_E<t$, there are two possible resource supply patterns, denoted A and B, that can lead to the minimum resource supply. We now present these two patterns and the corresponding $\mbox{sbf}$ equations, $\mbox{sbf}_{\mathbb{I}}(t_E, t)_{[A]}$ and $\mbox{sbf}_{\mathbb{I}}(t_E, t)_{[B]}$.

\textbf{Pattern A:} $s_1=T-C^L$. The scenario of pattern A  is shown in Figure~\ref{fig:sbf_pattern_2}, where $n_E=\left [ \left \lfloor \frac{t_E-(T-C^L)}{T} \right \rfloor \right ]_0$ and $n = \left [ \left \lfloor \frac{t-(T-C^L)}{T} \right \rfloor \right ]_0$. In the first period, $C^L$ units of resource are supplied as early as possible and hence during $[0,2\times (T-C^L)]$, no resource is supplied. In the following periods until time instant $s_E (= n_E \times T + T - C^L)$, $C^L$ units are supplied as late as possible. In the period $[s_E,e_E]$, the amount of supply depends on the distance of $t_E$ from $s_E$. If $t_E - s_E < C^L$, then the resource supply in this period cannot be exhausted when component $\mathbb{C}$ has EMS at $t_E$. Therefore interface $\mathbb{I}$ will provide $C^H$ units of resource in this period, because it can signal its mode switch to the parent component. On the other hand, if $t_E - s_E \geq C^L$ as in the example figure, then the resource supply in $[s_E,e_E]$ can be exhausted before component $\mathbb{C}$ experiences EMS, and hence the interface may only provide $C^L$ units in this period. After time instant $e_E$, the interface is guaranteed to provide $C^H$ units of resource in every period. An important boundary case to consider is when $e = n \times T + 2T - C^L = e_E$ and $t_E-s_E \geq C^L$. That is, when $t_E$ and $t$ are in the same period and the interface can exhaust its resource supply before EMS of component $\mathbb{C}$ (scenario shown in Figure~\ref{fig:sbf_pattern_boundary}). In this case, the minimum supply in this period can happen when it is provided as late as possible (for instance when $e - t > C^H - C^L$). We record the $\mbox{sbf}$ corresponding to the pattern of Figures~\ref{fig:sbf_pattern_2} and~\ref{fig:sbf_pattern_boundary} below.
{\small
\begin{align}
\nonumber & \mbox{sbf}_{\mathbb{I}}(t_E, t)_{[A]} = \\
& \begin{cases}
n_E \times C^L + (n - n_E) \times C^H \\
+ \left[t - (2T-C^L-C^H) - n \times T \right]_0 & t_E - s_E < C^L \\
&\\
(n_E + 1) \times C^L + (n - n_E - 1) \times C^H \\
+ \left[t - (2T-C^L-C^H) - n \times T \right]_0 & e \neq e_E \wedge \\ & t_E - s_E \geq C^L \\
&\\
n_E \times C^L \\
+ \min \left \{ C^L, \left[t - (2T-C^L-C^H) - n \times T \right]_0 \right \}  & e = e_E \wedge \\ & t_E - s_E \geq C^L
\end{cases} \label{eqn:sbf_pattern_2}
\end{align}
}

\textbf{Pattern B}: $s_1=T-C^L-(x_E-t_E)$, where $x_E = \left \lceil \frac{t_E}{T} \right \rceil \times T$.  Scenario of pattern B is shown in Figure~\ref{fig:sbf_pattern_3}, which is obtained by shifting pattern A in Figure~\ref{fig:sbf_pattern_2} by $x_E - t_E$. In this case,
{\small
\begin{align*}
& n_E = \left [ \left \lfloor \frac{t_E-s_1}{T} \right \rfloor \right ]_0, n = \left [ \left \lfloor \frac{t-s_1}{T} \right \rfloor \right ]_0, e_E = t_E - C^L + T \\
& \mbox{and } e = n \times T  + T + s_1.
\end{align*}
}
The $\mbox{sbf}$ corresponding to this shifted supply pattern is given below. It is similar to the previous case, except that the interface period containing $t_E$ is now guaranteed to supply no more than $C^L$ time units.
{\small
\begin{align}
\nonumber & \mbox{sbf}_{\mathbb{I}}(t_E, t)_{[B]} = \\
& \begin{cases}
(n_E + 1) \times C^L + (n - n_E - 1) \times C^H \\
+ \left[t - s_1 - (T-C^H) - n \times T \right]_0 & e \neq e_E \\
&\\
n_E \times C^L \\
+ \min \left \{ C^L, \left[t - s_1 - (T-C^H) - n \times T \right]_0 \right \}  & e = e_E \\
\end{cases} \label{eqn:sbf_pattern_3}
\end{align}
}

The following lemma proves that it is sufficient to consider the above two supply patterns for determining the $\mbox{sbf}$. 
\begin{lemma}
\label{lemma:sbfcase}
When $t_E < t$, pattern A or B are the only two possible supply patterns that can result in the minimal resource supply from interface $\mathbb{I}$.
\end{lemma}
\begin{proof}
\ae{
Suppose there exists a $s_1\in[0,T)$ such that $s_1\neq T-C^L$ (pattern A) and $s_1\neq T-C^L-(x_E-t_E)$ (pattern B), but $s_1$ leads to the minimal supply pattern for  time interval length $t$. \textbf{Case 1} ($s_1=T-C^L+\epsilon| 0<\epsilon\leq C^L$): In this case,  it is easy to see that the supply will be greater than or equal to $\mbox{sbf}_{\mathbb{I}}(t_E, t)_{[A]}$, because the supply for the first interface period will increase by $\epsilon$ and the supply for the last interface period will decrease by at most $\epsilon$. \textbf{Case 2} ($s_1=T-C^L-(x_E-t_E)+\epsilon| 0<\epsilon<(x_E-t_E)$): In this case, the supply for the interface period containing $t_E$ will stay the same or increase  by $C^H-C^L$ while the  supply for the last interface period may decrease by at most $\epsilon$ compared with the case when $s_1=T-C^L-(x_E-t_E)$. Therefore this supply is also minimized when $\epsilon \rightarrow x_E-t_E$ or $\epsilon \rightarrow 0$.  \textbf{Case 3} ($s_1=T-C^L-(x_E-t_E)-\epsilon|0<\epsilon\leq T-C^L-(x_E-t_E)$): In this case, the supply for the interface period containing $t_E$ will stay the same, while the supply for the last interface period may stay the same or increase compared with the case when $s_1=T-C^L-(x_E-t_E)$. Therefore in this case as well, the supply  is minimized when $\epsilon \rightarrow 0$. Combining the above cases, we can conclude that the supply is  minimized with either pattern A or pattern B.}
\end{proof}
Thus, a safe lower bound for $\mbox{sbf}_{\mathbb{I}}$ for the case when $t_E < t$ can be stated as follows.
\begin{equation}
\label{eqn:sbfAB}
\mbox{sbf}_{\mathbb{I}}(t_E, t)=\min\left\{\mbox{sbf}_{\mathbb{I}}(t_E, t)_{[A]},~\mbox{sbf}_{\mathbb{I}}(t_E, t)_{[B]} \right\}
\end{equation}

\begin{figure}[tbp]
\begin{tikzpicture}[scale=0.65,transform shape]
\coordinate (A) at (0,0);
\coordinate (A1) at (0,1);
\coordinate (B) at (2,0);
\coordinate (B1) at (2,1);
\coordinate (C) at (4,0);
\coordinate (C1) at (4,1);
\coordinate (D) at (6,0);
\coordinate (D1) at (6,1);
\coordinate (E) at (8,0);
\coordinate (E1) at (8,1);
\coordinate (F) at (10,0);
\coordinate (F1) at (10,1);
\coordinate (F2) at (8.75,0);
\coordinate (G) at (12,0);
\coordinate (G1) at (12,1);

\coordinate (K) at (1.5,0);
\coordinate (L) at (3.25,0);
\coordinate (M) at (7.25,0);
\coordinate (N) at (8,0);
\coordinate (O) at (10.75,0);
\coordinate (O1) at (11.5,0);

\coordinate (R) at (7,0);
\coordinate (S) at (12.25,0);

\draw (A) -- (C);
\draw [->] (D) -- (12.5,0);
\node [draw=none] [right =0.5cm of C]{$\ldots$};

\draw [->] (A) -- (A1);
\draw [->] (B) -- (B1);
\draw [->] (C) -- (C1);
\draw [->] (D) -- (D1);
\draw [->] (E) -- (E1);
\draw [->] (F) -- (F1);
\draw [->] (G) -- (G1);

\node [rounded corners,draw=black,minimum width=0.75cm] [above right =0cm of A]{$C^L$};
\node [rounded corners,draw=black,minimum width=0.75cm] [above right =0cm of L]{$C^L$};
\node [rounded corners,draw=black,minimum width=0.75cm] [above right =0cm of M]{$C^L$};
\node [rounded corners,draw=black,minimum width=0.75cm] [above right =0cm of N]{$C^L$};
\node [rounded corners,draw=black,minimum width=0.75cm] [above right =0cm of O]{$C^L$};
\node [rounded corners,draw=black,pattern=grid,minimum height=0.5cm,minimum width=0.5cm] [above right =0cm of O1]{};

\node[draw=none] [below =0.2cm of K]{$\mathbf{0}$};
\node[draw=none] [below =0.2cm of F2]{$\mathbf{t_E}$};
\node[draw=none] [below =0.2cm of S]{$\mathbf{t}$};
\node[draw=none] [below =0.05cm of F]{$e_E$};
\node[draw=none] [below =0.05cm of G]{$e$};

\begin{scope}[on background layer]
\draw [ultra thick] (1.5,-0.2) -- (1.5,1);
\draw [ultra thick] (8.75,-0.2) -- (8.75,1);
\draw [ultra thick] (12.25,-0.2) -- (12.25,1);
\end{scope}

\draw [<->] (1.5,1.2) to node [above]{Interval of interest} (12.25,1.2);

\draw [<->] (0.75,-1) to node [below]{{\small $x_E-t_E$}} (1.5,-1);
\draw [<->] (1.5,-0.15) to node [below]{} (2,-0.15);
\draw [<-] (1.75,-0.17) to node [below right]{$s_1$} (2,-0.5);

\draw (0.75,-0.1) -- (0.75,-1.1);
\draw (1.5,-0.7) -- (1.5,-1.1);
\draw (2,0) -- (2,-0.2);
\draw (10,0) -- (10,-0.1);
\end{tikzpicture}
\caption{MCPR worst-case resource supply pattern~B}
\label{fig:sbf_pattern_3}
\end{figure}

\subsection{Interface Generation}
\label{sec:interface}
In this section we use the $\mbox{sbf}$, together with the $\mbox{dbf}$ of component $\mathbb{C}$, to generate interface $\mathbb{I}$. For component $\mathbb{C}$ to be schedulable using interface $\mathbb{I}$, it is sufficient to ensure that $\mbox{dbf}(\mathbb{C},t,t_E,t_I) \leq \mbox{sbf}_{\mathbb{I}}(t_E,t)$ for various time interval lengths. Below we first present the schedulability test for the case when component $\mathbb{C}$ does not experience EMS. That is, the interface only executes in $LC$ mode supplying $C^L$ resource capacity periodically.
\begin{theorem}
A mixed-criticality component  $\mathbb{C}$ is schedulable  in $LC$ mode with $\mbox{sbf}_{\mathbb{I}}(t_E=t, t)$ if,
$\forall t: 0 \leq t \leq t_{MAX}, \forall t_I: 0 \leq t_I \leq t$, \\
\begin{equation}	
\mbox{dbf}(\mathbb{C},t,t_E,t_I)  \leq \mbox{sbf}_{\mathbb{I}}(t_E=t, t)~\mbox{ If }t_E=t
\label{eqn:thm_lc}
\end{equation}
where $t_{MAX}$ is a pseudo-polynomial in the size of the input that can be derived using similar techniques in Section~\ref{sec:tmax} in the appendix, and $\mbox{sbf}_{\mathbb{I}}(t_E, t)$ is given by Equation~\eqref{eqn:sbf_pattern_1} in Section~\ref{sec:mcpr_sbf}.
\label{thm:thm_lc}
\end{theorem}
For a given $t$ and $t_I$, $\mbox{dbf}(\mathbb{C},t,t_E,t_I)$ can be computed using techniques described in Section~\ref{sec:dbfcomponent}. The only unknown quantity in Equation~\eqref{eqn:thm_lc} is the $LC$ resource capacity $C^L$. This capacity can then be computed exactly using existing techniques~\cite{EAL07}.

To compute the $HC$ resource capacity $C^H$, we need to consider the schedulability test when component $\mathbb{C}$ experiences EMS at some time instant $t_E (< t)$. The following theorem presents this test.
\begin{theorem}
A mixed-criticality component  $\mathbb{C}$ is schedulable in $HC$ mode with $\mbox{sbf}_{\mathbb{I}}(t_E, t)$ if
$\forall t: 0 \leq t \leq t_{MAX}, \forall t_E: 0 \leq t_E \leq t, \forall t_I: 0 \leq t_I \leq t_E$, 
\begin{equation}
\mbox{dbf}(\mathbb{C},t,t_E,t_I)  \leq \mbox{sbf}_{\mathbb{I}}(t_E, t) 
\label{eqn:thm_hc}
\end{equation}
where $\mbox{sbf}_{\mathbb{I}}(t_E, t)$ is given by Equation~\eqref{eqn:sbfAB} in Section~\ref{sec:mcpr_sbf}.
\label{thm:thm_hc}
\end{theorem}
The only unknown quantity in Equation~\eqref{eqn:thm_hc} is the $HC$ resource capacity $C^H$, assuming we have already computed $C^L$ using Theorem~\ref{thm:thm_lc}. $C^H$ can then be computed similar to $C^L$ using existing techniques~\cite{EAL07}.

\section{Evaluation}
\label{sec:evaluation}
In this section we evaluate the performance of the proposed mechanism in terms of offline schedulability as well as its ability to support $LC$ task executions online. Tasksets are generated using the following settings, where each parameter is randomly drawn from the given range using an uniform distribution.
\begin{itemize}
\itemsep=-0.3pt
\label{set:1}
	\item $u_i^L = C_i^L/T_i$ is in the range $[0.02,0.1]$.
	\item $C_i^H/C_i^L$ is in the range $[2,3]$.
	\item $T_i$ is in the range $[10,150]$.
  \item $D_i=T_i$ as service adaption strategy, one of the mechanisms being compared, can only support implicit deadline tasks.
	\item Task $\tau_i$ is deemed to be a $HC$ task with probability $0.5$.
	\item For a $HC$ task $\tau_i$, $D_i^L$ is determined by the deadline tuning algorithm in \cite{Eas13}.
  \item For the proposed mechanism, we assume that all the $|H|$ $HC$ tasks in the generated taskset are allocated to a $HC$ component $\mathbb{C}_H=\{\mathcal{W}_H,TL_H\}$, and all the $LC$ tasks are allocated to a $LC$ component $\mathbb{C}_L$.
\end{itemize}
We have chosen relatively small values for $u_i^L$ and $C_i^H/C_i^L$ so that sufficient number of $HC$ tasks are generated. This enables us to evaluate the online performance of various approaches when different number of $HC$ tasks synchronously switch to $HC$ mode. The generated taskset is evaluated for offline schedulability as well as online performance in terms of support for $LC$ execution  under four different mechanisms. These include the mechanism presented in this paper (``Proposed Mechanism''), service adaptation strategy~\cite{PCH14} (``Service Adaptation''), Interference Constraint Graph~\cite{PCH13} (``ICG''), and the classical mixed-criticality studies in which all the $LC$ jobs are dropped at the moment any $HC$ job switches to $HC$ mode~\cite{Eas13} (``Classical Model''). Note that the classical model can be obtained by setting $TL_H=0$ in our mechanism. In Section~\ref{sec:offline} we present our results for offline performance based on schedulability tests, and in Section~\ref{sec:online} we compare their online performance through simulations. 
\subsection{Offline Schedulability}
\label{sec:offline}
In order to generate feasible tasksets, we consider different  bounds for the term $\max \{ U_L^L + U_H^L, U_H^H \}$, where $U^L_L =\sum\limits_{L_i=LC}{C_i^L}/{T_i}$, $U_H^L = \sum\limits_{L_i=HC}{C_i^L}/{T_i}$ and $U_H^H = \sum\limits_{L_i=HC}{C_i^H}/{T_i}$. For each bound value, we generate $1000$ tasksets based on the procedure described above, and evaluate their off-line schedulability. For the elastic model~\cite{Su13} in which the $LC$ task periods are extended, any generated taskset with $U_H^H$ is always schedulable, because in the worst-case all the $LC$ task periods can be extended to infinity.  The schedulability test for the service adaption strategy~\cite{BBA12}  is a utilization based test. ICG uses the well known Audsley's algorithm to assign task priorities, and its schedulablity is maximized when the interference graph is fully connected, i.e., each $HC$ task has an execution dependency with every $LC$ task in the system. For our mechanism, if a hierarchical scheduling framework is considered, then we assume that the MCPR interface period $T$ for both  $\mathbb{C}_H$ and  $\mathbb{C}_L$ is equal to $5$ time units. This is reasonable because the smallest task period in any taskset is $10$ time units.

\begin{figure*}[t]
  \centering
  \begin{minipage}[t]{.32\linewidth}
    \includegraphics[width=2.2in, height=2.0in]{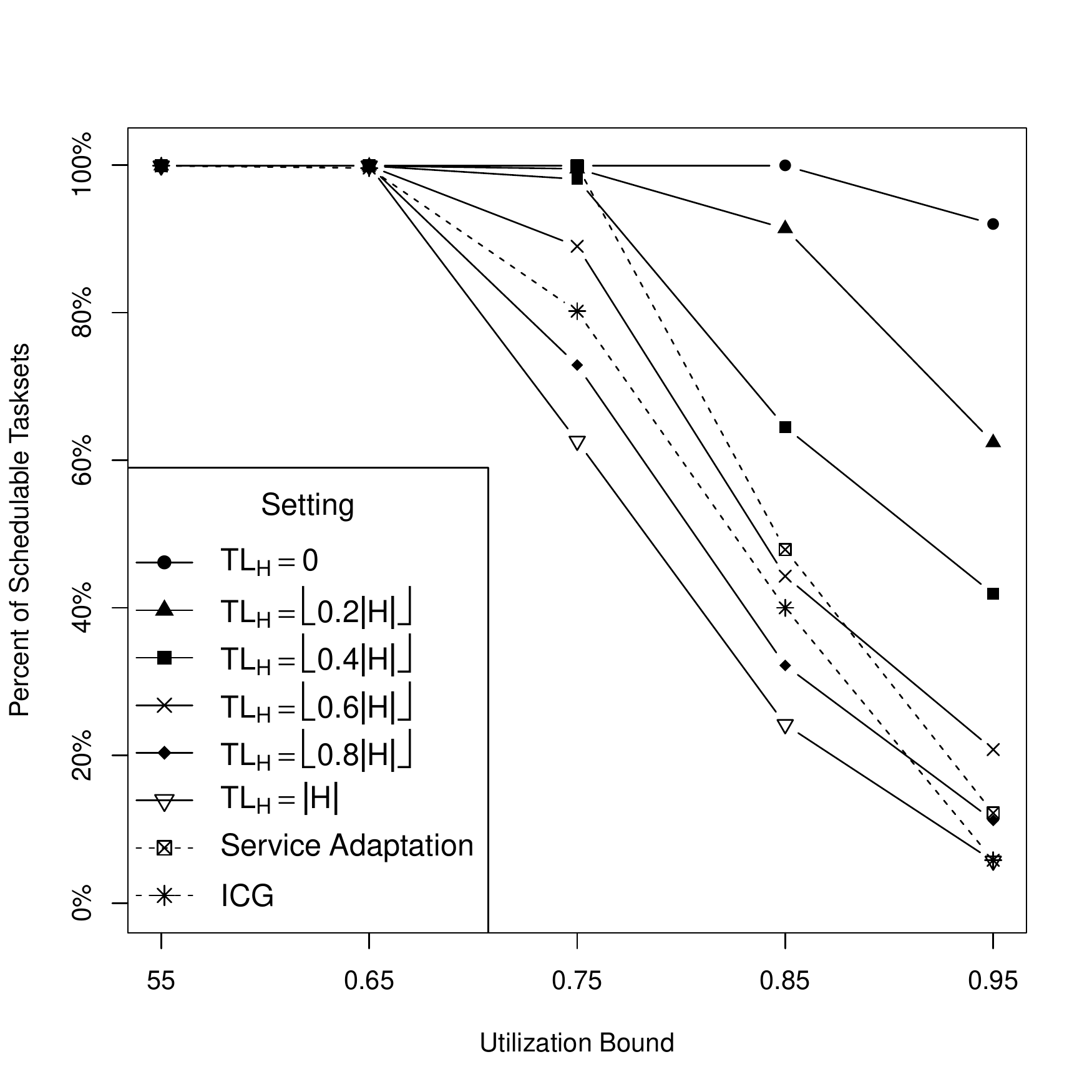}
  \caption{Schedulability under a Flat Scheduling Framework}
      \label{fig:subfig:a}
  \end{minipage}
\begin{minipage}[t]{.32\linewidth}
  \includegraphics[width=2.2in, height=2.0in]{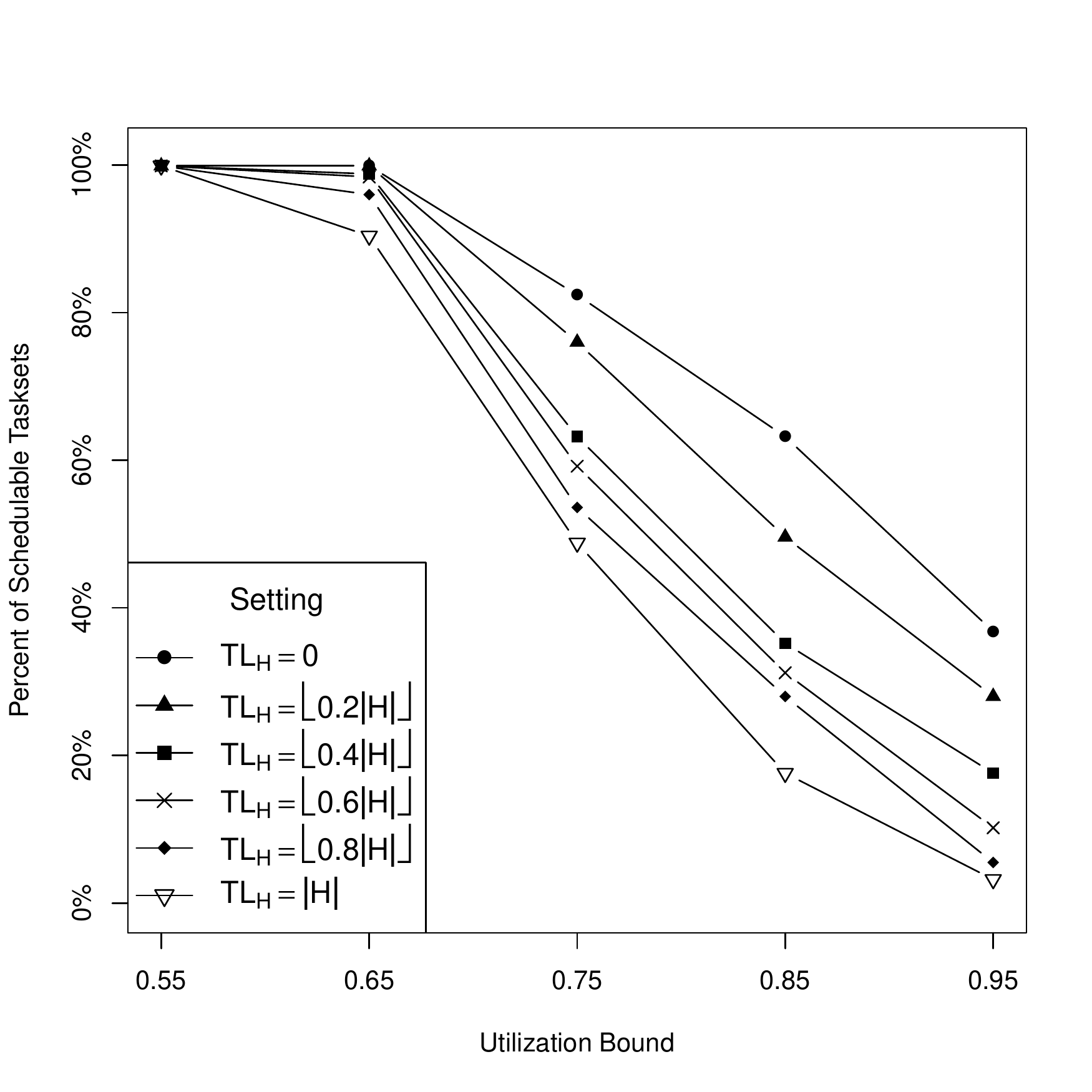}
  \caption{Schedulability under a Hierarchical Scheduling Framework}
    \label{fig:subfig:b}
  \end{minipage}
    \begin{minipage}[t]{.32\linewidth}
    \includegraphics[width=2.2in, height=2.0in]{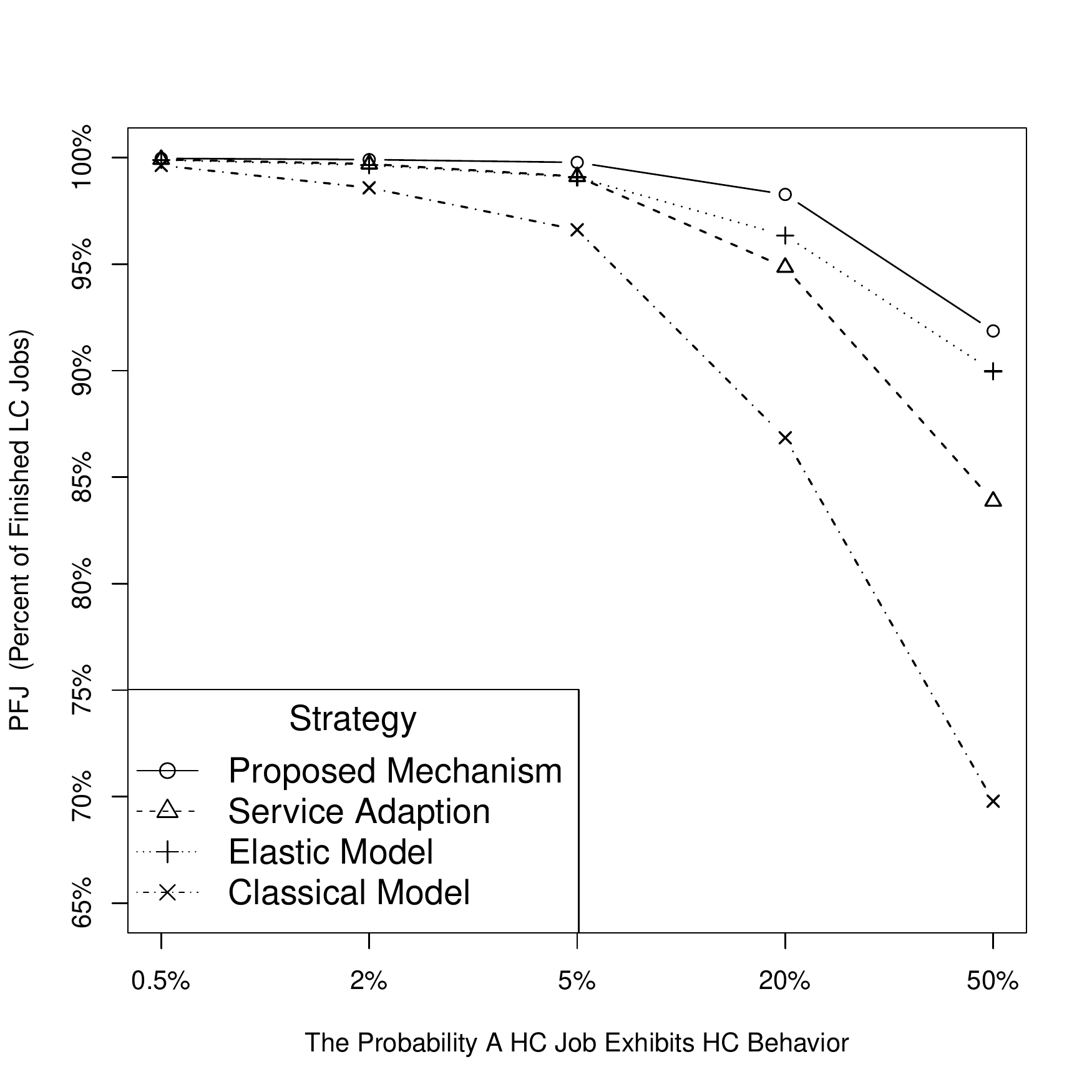}
  \caption{$\max\{ U^L_L + U_H^L, U_H^H \} = 0.8$}
\label{fig:p1}
  \end{minipage}
\end{figure*}

Figures~\ref{fig:subfig:a} and \ref{fig:subfig:b} show the schedulability performance for the tasksets  under various mechanisms. In Figures~\ref{fig:subfig:a} we present results for our mechanism under a flat scheduling framework, and in In Figures~\ref{fig:subfig:a} we present results for our mechanism under a hierarchical scheduling framework. In these figures, the x-axis denotes the bound value for $\max \{ U_L^L + U_H^L, U_H^H \}$, and the y-axis denotes schedulability ratio, i.e., percentage of tasksets deemed schedulable by the different mechanisms. For our mechanism, we generate the schedulability results for various values of the tolerance limit: $TL_H=0,~\lfloor 0.2|H|\rfloor,~\lfloor 0.4|H|\rfloor,~\lfloor 0.6|H|\rfloor,~\lfloor 0.8|H|\rfloor\mbox{ and } |H|$.

As shown in Figure~\ref{fig:subfig:a}, the schedulability performance of our mechanism clearly depends on the tolerance limit; a higher limit generally implies lower schedulability, because it uses additional resources to support $LC$ executions. For values of $TL_H$ up to $\lfloor 0.4|H|\rfloor$, our mechanism outperforms both service adaptation and ICG on an average. Similar trends can also be observed for our mechanism under a hierarchical framework, except that the schedulability drops more rapidly due to the overhead of hierarchical scheduling. The classical model is represented by the curve with $TL_H=0$ and it has the highest schedulability, but offers no support for $LC$ executions when $HC$ tasks switch to $HC$ mode. Thus we can conclude that as long as no more than $\lfloor 0.4|H|\rfloor$ of the $HC$ tasks execute in $HC$ mode at each time instant, our mechanism offers the best performance in terms of offline schedulability as well as online support for $LC$ executions.


\subsection{Online Support for $LC$ Executions}
\label{sec:online}
In this section, we compare the performance of our mechanism in terms of its ability to support $LC$ executions with the other mechanisms described above. We use the following quantitative parameter to measure this online $LC$ performance.
\begin{definition}[Percentage of Finished $LC$ Jobs (\textbf{$PFJ$})]
Let $MAX_t$ denote the maximum possible number of $LC$ jobs that a taskset $\mathcal{T}$ can generate in the time interval $[0, t)$. By definition, $MAX_t = \sum\limits_{L_i=LC} \lceil {t}/{T_i}\rceil$. Let $FIN_t$ denote the number of $LC$ jobs that successfully finish by their deadlines in the time interval $[0, t)$ using some mechanism. Then, $PFJ$ is equal to ${FIN_t}/{MAX_t}$.
\end{definition}
\begin{figure}[tbp]
\centering
\includegraphics[scale=0.335]{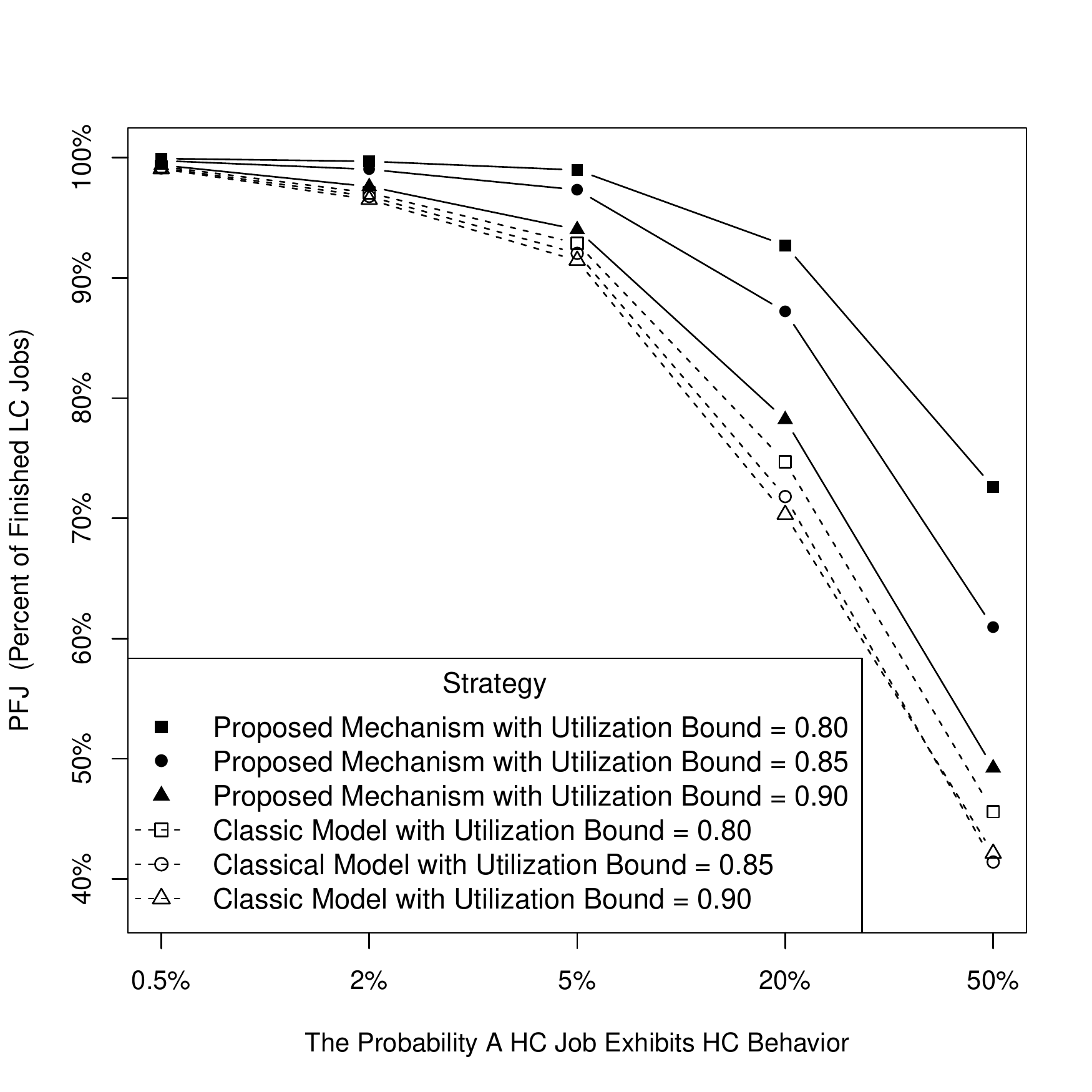}
\caption{$\max \{ U^L_L + U_H^L, U_H^H \} = 0.8,0.85 \mbox{ and }0.9$}
\label{fig:p4}
\end{figure}
Tasksets are generated using the procedure described earlier, and the various mechanisms are simulated to measure their online performance. The following additional settings and restrictions are used for this purpose. 
\begin{itemize}
\itemsep=-0.8pt
	\item $\max \{ U^L_L + U_H^L, U_H^H \} = 0.8,0.85 \mbox{ and }0.9$. 
 	\item Tolerance limit $TL_H$ is chosen to be the largest value that still guarantees schedulability of our mechanism under a flat scheduling framework.
	\item Tasksets are simulated for $t=10,000$ time units.
	\item Each $HC$ job independently switches to $HC$ mode, i.e., executes for more than $LC$ WCET, with a probability of $0.005, 0.02, 0.05, 0.2$ or $0.5$. 
	\item All the mechanisms will transition back to $LC$ mode of execution when there are no pending jobs.
\end{itemize}

We have chosen a relatively high value for $\max \{ U^L_L + U_H^L, U_H^H \}$, because at smaller values there is sufficient spare capacity so that all the mechanisms are easily able to support $LC$ executions. Simulation results are shown in Figures~\ref{fig:p1} and \ref{fig:p4}. The x-axis denotes the probability that a $HC$ job independently switches to $HC$ mode, and the y-axis denotes $PFJ$ for each mechanism. Each point in these figures is generated by taking an average value of $PFJ$ over $1000$ tasksets. In Figure~\ref{fig:p1}, we consider only those tasksets that are deemed to be offline schedulable by all the presented mechanisms. As shown in the figure, our mechanism consistently outperforms all the other mechanisms for different values of mode switch probability, and the performance gap improves with increasing probability values.
\ae{
One should note that the results in Figure~\ref{fig:p1} may not be truly representative of the performance of our mechanism in terms of its ability to support $LC$ jobs, and this can be explained as follows. To compare our mechanism's ability to support $LC$ executions with the other mechanisms, we have to simulate using tasksets that are schedulable by all these mechanisms.  In particular, it does not include many tasksets that are schedulable under our mechanism, but not under one of the other mechanisms.  From our observation, in the tasksets that are  schedulable by all these mechanisms, the average percentage of $HC$ tasks is much higher than that of $LC$ tasks. Hence to show the ability of our mechanism to support $LC$ executions in  a  more objective way,  we compare  the proposed mechanism alone with the classical model,  with utilization bound $\max \{ U^L_L + U_H^L, U_H^H \}=0.8,~0.85 \mbox{ and } 0.9$ as shown in Figure~\ref{fig:p4}. In this case, any taskset schedulable  by the classical model can be used in  the simulation. It can be seen that the performance of both our mechanism and the classical model drops when compared with the results in Figure~\ref{fig:p1}. However, it can also been seen that, our mechanism still dominates the classical model and the corresponding performance gap does not decrease compared with the gap in Figure~\ref{fig:p1}.
}

\section*{Acknowledgment}
\ae{
This work was supported in part by MoE Tier-2 grant (MOE2013-T2-2-029) and NTU start-up grant, Singapore. This work was also supported in part by MSIP/IITP (14-824-09-013) funded by the Korea Government.}

\section{Conclusions}
\label{sec:conclusions}
In this paper we proposed a novel mechanism to improve the service levels of low-criticality tasks by  allowing them to execute even when some high-criticality tasks have exceeded their estimated WCETs. We developed schedulability tests for our mechanism under the mixed-criticality EDF scheduling strategy, considering both a flat as well as an hierarchical scheduling framework. We also evaluated the performance of our mechanism in terms of offline schedulability and online support for low-criticality executions. Simulation results clearly show that the proposed mechanism outperforms all the existing approaches.

In the evaluation section we only consider the performance of our mechanism when all the high-criticality tasks are in one component and all the low-criticality tasks are in another component. In fact, its performance can be further improved if we also consider scenarios in which the low-criticality tasks are allocated to the same component as the high-criticality ones, especially in terms of offline schedulability. In our future work we will consider this problem of optimally allocating the low-criticality tasks so as to maximize offline schedulability as well as online performance.


\bibliographystyle{IEEEtran}
\bibliography{all}
\appendix
\section{Appendix}
\subsection{$\mbox{Dbf}$ Optimization}
\label{sec:opt}
When component $\mathbb{C}$  experiences EMS, i.e., the case when $t_E<t$, it is pessimistic to simply add up the demand of all the tasks. Here we introduce an optimization that can be applied in the schedulability test to reduce this pessimism.  We split $\mbox{dbf}(\tau_i,t,t_i)$ into two elements, $\mbox{DL}(\tau_i,t,t_i)$ denoting the demand for the interval $[0,t_E)$, and $\mbox{DH}(\tau_i,t,t_i)$ denoting the demand for the interval $[t_E,t)$. 
\begin{small}
\begin{equation}
dbf(\tau_i,t,t_i) = \mbox{DL}(\tau_i,t,t_i) +\mbox{DH}(\tau_i,t,t_i)
\end{equation}
\end{small}

Below we present a key observation that provides some insight into this split. Since the first deadline miss is assumed to happen at time instant $t$ in our schedulability test, the demand before $t_E|<t$ cannot exceed $t_E$. Otherwise, the first deadline miss would happen at or before $t_E$. Thus the total demand during $[0,t_E)$ can be bounded by $t_E$, and as a consequence $\mbox{dbf}(\mathbb{C},t,t_E,t_I)$ can be more tightly bounded as follows.
\begin{small}
\begin{equation}
\label{eqn:optidbf}
\begin{split}
&\mbox{dbf}(\mathbb{C},t,t_E,t_I)=\mbox{DL}+\mbox{DH}+\sum\limits_{\Delta_i\in \mathcal{G}} \Delta_i \mbox{, where}\\
&\mbox{DH}=\sum\limits_{L_i=LC} \mbox{DH}(\tau_i,t,t_I)+\sum\limits_{L_i=HC} \mbox{DH}(\tau_i,t,t_E), \mbox{ and} \\
&\mbox{DL}=\min \left \{t_E, \sum\limits_{L_i=LC} \mbox{DL}(\tau_i,t,t_I)+\sum\limits_{L_i=HC} \mbox{DL}(\tau_i,t,t_E) \right \} \\
\end{split}
\end{equation}
\end{small}

In Equation~\ref{eqn:optidbf}, we use $\mbox{DL}$ to bound the total demand  of $\mathbb{C}$ for the interval $[0,t_E)$, and $\mbox{DH}$ to bound the total demand for  the interval $[t_E, t)$. In order to maximize the total demand, we must then split the demand between $\mbox{DL}$ and $\mbox{DH}$ such that $\mbox{DH}$ is maximized (or equivalently $\mbox{DL}$ is minimized). This is because the total demand for the interval $[0, t_E)$ is bounded by $t_E$.

In Section~\ref{sec:task} we already present $\mbox{dbf}(\tau_i,t,t_i)$ when task $\tau_i$ satisfies condition $a$, $b$, $c$ or $d$. Here we present $\mbox{DL}(\tau_i,t,t_i)$ and  $\mbox{DH}(\tau_i,t,t_i)$ for these cases, such that $\mbox{DH}(\tau_i,t,t_i)$ is maximized. If $\tau_i$ is a $LC$ task, then it cannot execute after $t_E$ (dropped at $t_i=t_I\leq t_E$). Hence for \textbf{condition a},  \\
\begin{small}
\begin{equation}
\label{eqn:sl}
\begin{split}
&\mbox{DL}(\tau_i,t,t_I)_{[a]}=\mbox{dbf}(\tau_i,t,t_I)_{[a]}\\ 
&\mbox{DH}(\tau_i,t,t_I)_{[a]}=0
\end{split}
\end{equation}
\end{small}

Consider the case when $\tau_i$ satisfies \textbf{condition b}, i.e.,  $L_i=HC$ and $t-t_i< D_i-D_i^L$. Here as well $\tau_i$ cannot execute after $t_E$ as given in Lemma~\ref{lemma:HCcase1}. Hence,
\begin{small}
\begin{equation}
\label{eqn:sh1}
\begin{split}
&\mbox{DL}(\tau_i,t,t_E)_{[b]}=\mbox{dbf}(\tau_i,t,t_E)_{[b]}\\
&\mbox{DH}(\tau_i,t,t_E)_{[b]}=0
\end{split}
\end{equation}
\end{small}

\begin{figure}[htbp]
\centering
\includegraphics[scale=0.7]{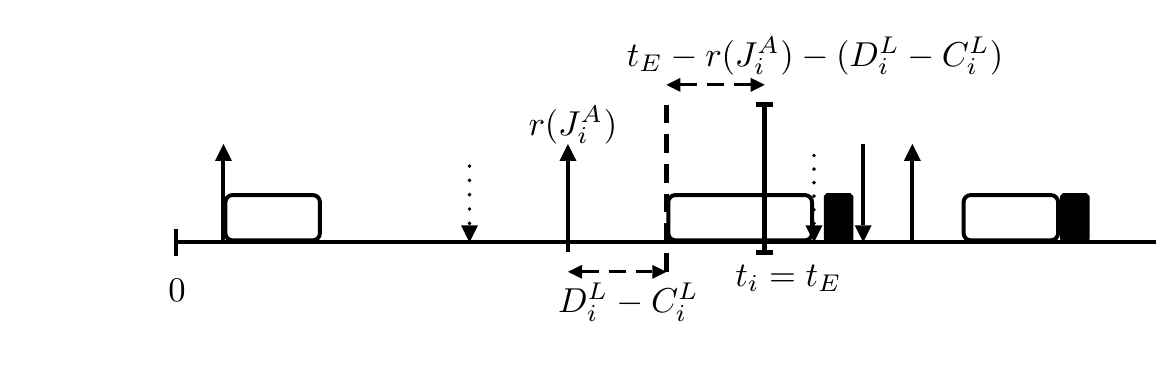} 
\caption{$\mbox{DL}(\tau_i,t,t_E)_{[c]}$ and $\mbox{DH}(\tau_i,t,t_E)_{[c]}$}
  \label{fig:HC1}
\end{figure}

Consider the case when $\tau_i$ satisfies \textbf{condition c}, i.e., $L_i=HC$ and $t-t_i \geq D_i$. In this case $t_i(=t_E)$ occurs after the release of special job $J_i^A$ and this scenario is shown in Figure~\ref{fig:HC1}.  To minimize the demand of $J_i^A$ before $t_E$, we assume that it executes as late as possible. Thus, $J_i^A$'s demand before $t_E$ can be bounded by $t_E-r(J_i^A)-(D_i^L-C_i^L)$, and we have,  \\
\begin{small}
\begin{equation}
\label{eqn:sh2}
\begin{split}
&\mbox{DL}(\tau_i,t,t_E)_{[c]}\!=\!\min\left\{\left[t_E-r(J_i^A)-(D_i^L-C_i^L)\right]_0,C_i^L \right\}\\&~~~~~~~~~~~~~~~~~~~~~~+b_i\times C_i^L, \\
&\mbox{DH}(\tau_i,t,t_E)_{[c]}\!=\!-\!\min\left\{\left[t_E-\!r(J_i^A)-\!(D_i^L-C_i^L)\right]_0,C_i^L  \right\}\\&~~~~~~~~~~~~~~~~~~~~~~+\mbox{dbf}(J_i^A,t,t_E)+a_i\times C_i^H, \mbox{ where} \\ 
&b_i=\left \lfloor \left(t_E-\!(t-D_i-\!\lfloor {(t-D_i)}/{T_i} \rfloor \times T_i)\right)/T_i \right\rfloor,\\
&a_i = \left\lfloor \left(t-D_i\right)/T_i\right \rfloor-b_i,\mbox{~and}\\
&r(J_i^A)=t-D_i-\lfloor {(t-D_i)}/{T_i} \rfloor \times T_i+b_i\times T_i.
\end{split}
\end{equation}
\end{small}

Finally, consider the case when $\tau_i$ satisfies \textbf{condition d}, i.e., $L_i=HC$ and $D_i-D_i^L \leq t-t_i < D_i$. In this case as well $\mbox{DH}(\tau_i,t,t_i=t_E)_{[d]}$ is maximized if the first job is released at $t-D_i-\lfloor {(t-D_i)}/{T_i} \rfloor \times T_i$ (pattern of condition c), and therefore we have, \\
\begin{small}
\begin{equation}
\begin{split}
&\mbox{DH}(\tau_i,t,t_E)_{[d]}=\mbox{DH}(\tau_i,t,t_E)_{[c]}\\
&\mbox{DL}(\tau_i,t,t_E)_{[d]}=\mbox{dbf}(\tau_i,t,t_E)_{[d]}-\mbox{DH}(\tau_i,t,t_E)_{[d]}
\end{split}
\end{equation}
\end{small}

\subsection{Upper bound for $t_{MAX}$}
\label{sec:tmax}
Consider a mixed-criticality system with $p$ $HC$ components $\mathbb{C}_{1},\mathbb{C}_{2},\ldots,\mathbb{C}_{p}$ and $q$ $LC$ components $\mathbb{C}_{p+1},\mathbb{C}_{p+2},\ldots,\mathbb{C}_{p+q}$. Let $U_L^L(j)=\sum\limits_{Li=LC}^{\tau_i\in \mathbb{C}_j}C_i^L/T_i$, $U_H^L(j)=\sum\limits_{Li=HC}^{\tau_i\in \mathbb{C}_j}C_i^L/T_i$ and $U_H^H(j)=\sum\limits_{Li=HC}^{\tau_i\in \mathbb{C}_j}C_i^H/T_i$.

\textbf{Case 1}: If  component $\mathbb{C}_j$ experience IMS at  $t_{Ij}$, then the demand of a $LC$ task $\tau_i$  in the time interval $[0,t)$ is  upper bounded by $(t_{Ij}/T_i+1)\times C_i^L$, because $\tau_i$ will be dropped after $t_{Ij}$.  

A $HC$ task  $\tau_i$ in  $\mathbb{C}_j$ switches to $HC$ mode at some time instant $t_i\in [t_{Ij},t_E]$. The demand of $\tau_i$ before job $J_i^A$ is bounded by ${t_i}/{T_i}\times C_i^L$, the demand of job $J_i^A$ is bounded by $C_i^H$, and the demand after  $t_i$ is bounded by $(t-t_i-D_i+T_i)/{T_i}\times C_i^H$. Thus the total demand of $\tau_i$ in the time interval $[0,t)$ is bounded by 
\begin{eqnarray}
\label{eqn:h1case1}
\frac{t_i}{T_i}\times C_i^L+C_i^H+\frac{t-t_i-D_i+T_i}{T_i}\times C_i^H 
\end{eqnarray}
Since $C_i^H>C_i^L$ and $t_i\in[t_{Ij},t_E]$,  the value of Expression (\ref{eqn:h1case1}) is maximized when $t_i=t_{Ij}$. Therefore  the  total demand  of $\mathbb{C}_j$  is bounded by 
\begin{small}
\begin{align*}
&\sum\limits_{Li=HC}^{\tau_i\in \mathbb{C}_j}\left(t_{Ij}\times C_i^L+C_i^H(t-t_{Ij}-D_i+2T_i)\right)/T_i\\
&+\sum\limits_{L_i=LC}^{\tau_i\in \mathbb{C}_j}(t_{Ij}/T_i+1)\times C_i^L  \\
&\leq U_H^H(j)\times t+\max\limits_{\tau_i\in\mathbb{C}_j }\{2T_i-D_i\}\times U_H^H(j)+\sum\limits_{L_i=LC}^{\tau_i\in \mathbb{C}_j} C_i^L\\&+(U_L^L(j)+U_H^L(j)-U_H^H(j))\times t_{Ij}
\end{align*}
\end{small}
\textbf{Case 2}: Suppose  component $\mathbb{C}_j$ does not experience IMS, i.e., all the $LC$ tasks within $\mathbb{C}_j$ are dropped after $t_E$, and all the $HC$ tasks  switch to $HC$ mode at $t_E$. In this case,  the demand of a $LC$ task $\tau_i$  in the time interval $[0,t)$ is  upper bounded by $(t_E/T_i+1)\times C_i^L$, and the demand of a $HC$ task $\tau_i$  in the time interval $[0,t)$ is  upper bounded by $\frac{t_E}{T_i}\times C_i^L+C_i^H+\frac{t-t_E-D_i+T_i}{T_i}\times C_i^H$. Therefore  the  total demand of $\mathbb{C}_j$  is bounded by 
\begin{small}
\begin{align*}
&\sum\limits_{Li=HC}^{\tau_i\in \mathbb{C}_j}\left(t_E\times C_i^L+C_i^H\times(t-t_E-D_i+2T_i)\right)/T_i \\&+\sum\limits_{L_i=LC}^{\tau_i\in \mathbb{C}_j}(t_E/T_i+1)\times C_i^L\\ & \leq U_H^H(j) \times t+\max\limits_{\tau_i\in \mathbb{C}_j}\{2T_i\!-D_i\}\times U_H^H(j)+\sum\limits_{L_i=LC} ^{\tau_i\in \mathbb{C}_j}C_i^L\\&~+(U_L^L(j)+U_H^L(j)-U_H^H(j))\times t_E
\end{align*}
\end{small}
Let A denote the set of components $\mathbb{C}_j$ with $U_L^L(j)+U_H^L(j)-U_H^H(j)<0$, and B denote the remaining set of components. Then if $\mathbb{C}_j\in A$, its demand bound given above is maximized when $t_{Ij}=0$ or $t_E=0$. On the other hand, if $\mathbb{C}_j\in B$, its demand bound is maximized when $t_{Ij}=t$ or $t_E=t$. Thus, an upper bound on the total demand of $\mathbb{C}_j$ is equal to
\begin{equation}
\label{eqn:carryover}
\begin{split}
\begin{cases}
&U_H^H(j) \times t+\max\limits_{\tau_i\in \mathbb{C}_j}\{2T_i-D_i\}\times U_H^H(j)\\&+\sum\limits_{L_i=LC} ^{\tau_i\in \mathbb{C}_j}C_i^L~~~~~~~~~~~~~~~~~~~~\mbox{if }\mathbb{C}_j\in A \\
&\max\limits_{\tau_i\in \mathbb{C}_j}\{2T_i-D_i\}\times U_H^H(j)+\sum\limits_{L_i=LC} ^{\tau_i\in \mathbb{C}_j}C_i^L\\&+(U_L^L(j)+U_H^L(j))\times t~~~~~~~\mbox{if }\mathbb{C}_j\in B
\end{cases}
\end{split}
\end{equation}
Suppose $\sum\limits_{j=1}^{j\leq p+q} \mbox{dbf}(\mathbb{C}_{j},t,t_E,t_{Ij})> t
$ for some $t$. Then it must be the case that 
\begin{small}
\begin{align*}
&\sum\limits_{j=1}^{j\leq p+q} \left(\max\limits_{\tau_i\in \mathbb{C}_j}\{2T_i-D_i\}\times U_H^H(j)+\sum\limits_{L_i=LC} ^{\tau_i\in \mathbb{C}_j}C_i^L\right)\\
&>t\left(1-\sum\limits_{\mathbb C_j\in A}U_H^H(j)-\sum\limits_{\mathbb C_j\in B}(U_L^L(j)+U_H^L(j))\right)\\
& \Rightarrow t<\frac{\sum\limits_{j=1}^{j\leq p+q} \left(\max\limits_{\tau_i\in \mathbb{C}_j}\{2T_i-D_i\}\times U_H^H(j)+\sum\limits_{L_i=LC} ^{\tau_i\in \mathbb{C}_j}C_i^L\right)}{1-\sum\limits_{\mathbb C_j\in A}U_H^H(j)-\sum\limits_{\mathbb C_j\in B}(U_L^L(j)+U_H^L(j))}
\end{align*}
\end{small}
Thus we can conclude that the upper bound of $t$, i.e., $t_{MAX}$, is given as  
\begin{small}
\begin{align*}
\frac{\sum\limits_{j=1}^{j\leq p+q} \left(\max\limits_{\tau_i\in \mathbb{C}_j}\{2T_i-D_i\}\times U_H^H(j)+\sum\limits_{L_i=LC} ^{\tau_i\in \mathbb{C}_j}C_i^L\right)}{1-\sum\limits_{\mathbb C_j\in A}U_H^H(j)-\sum\limits_{\mathbb C_j\in B}(U_L^L(j)+U_H^L(j))}
\end{align*}
\end{small}

\end{document}